\documentclass[11pt]{iopart}

\usepackage{bbm, changes}
\usepackage{amssymb,amsthm}
\usepackage{color, setstack}
\usepackage{graphicx,epsfig}
\usepackage{times}
\usepackage{tikz}
\usepackage{mathptmx}
\usepackage[10pt]{moresize}
\usepackage{chngcntr}
\counterwithout{figure}{section}
\theoremstyle{plain}

\newtheorem{thm}{Theorem}
\newtheorem{lem}[thm]{Lemma}
\newtheorem*{cor}{Corollary}
\theoremstyle{definition}

\newcommand{\ket}[1]{|#1\rangle}
\newcommand{\bra}[1]{\langle #1|}

\definecolor{myurlcolor}{rgb}{0,0,0.7}
\definecolor{myrefcolor}{rgb}{0.8,0,0}
\usepackage{hyperref}
\hypersetup{colorlinks, linkcolor=myrefcolor,
citecolor=myurlcolor, urlcolor=myurlcolor}

\definecolor{matty}{rgb}{1,0,0}
 
\begin{document}

\title[Self-testing protocols based on the chained Bell inequalities]{Self-testing protocols based on the chained Bell inequalities}

\author{I. \v{S}upi\'{c}$^1$, R. Augusiak$^1$, A. Salavrakos$^1$, A. Ac\'in$^{1,3}$}
\address{$^1$ ICFO-Institut de Ciencies Fotoniques, The Barcelona Institute of Science and Technology, 08860 Castelldefels (Barcelona), Spain}
\address{$^2$ Center for Theoretical Physics, Polish Academy of Sciences, Aleja Lotnik\'ow 32/46, 02-668 Warsaw, Poland}
\address{$^3$ ICREA--Institucio Catalana de Recerca i
Estudis Avan\c{c}ats, Lluis Companys 23, 08010 Barcelona, Spain}
\ead{ivan.supic@icfo.es}

\date{\today}

\vspace{10pt}

\begin{abstract}
Self-testing is a device-independent technique based on non-local correlations whose aim is to certify the effective uniqueness of the quantum state and measurements needed to produce these correlations. It is known that the maximal violation of some Bell inequalities suffices for this purpose. However, most of the existing self-testing protocols for two devices exploit the well-known Clauser-Horne-Shimony-Holt Bell inequality or
modifications of it, and always with two measurements per party. Here, we generalize the previous results by demonstrating that one can construct self-testing protocols based on the chained Bell inequalities, defined for two devices implementing an arbitrary number of two-output measurements. On the one hand, this proves that the
quantum state and measurements leading to the maximal violation of the chained Bell inequality are unique. On the other hand, in the limit of a large number of measurements, our approach allows one to self-test the entire plane of measurements spanned by the Pauli matrices $X$ and $Z$. Our results also imply that the chained Bell
inequalities can be used to certify two bits of perfect randomness.
\end{abstract}

% Uncomment for PACS numbers
%\pacs{00.00, 20.00, 42.10}
%
% Uncomment for keywords
%\vspace{2pc}
%\noindent{\it Keywords}: self-testing, chained Bell inequalities, SOS decomposition, robustness
%
% Uncomment for Submitted to journal title message
%\submitto{\JPA}
%
% Uncomment if a separate title page is required
%\maketitle
% 
% For two-column output uncomment the next line and choose [10pt] rather than [12pt] in the \documentclass declaration
%\ioptwocol
%

\section{Introduction}
In the last decades, it has been proven that nonlocality,
besides being very important from  a foundational
point of view, is also a resource for quantum information applications in the so-called device-independent scenario. There, devices are just seen as ``black boxes" producing a classical output, given a classical input. The devices can provide an advantage over classical information processing only when they produce non-local correlations, that is, correlations that violate a Bell inequality and, therefore, cannot be reproduced by shared classical instructions. It is then possible to construct quantum information applications exploiting this device-independent quantum certification based on Bell's theorem. Successful examples are protocols for device-independent randomness generation~\cite{DIRNG}, device-independent quantum key distribution~\cite{DIQKD} and blind quantum computation~\cite{BQC}.

Historically, self-testing can be considered as the first device-independent protocol. Introduced by Mayers and Yao~\cite{MayersYao}, the standard self-testing scenario consists of a classical user who has access to several black boxes, which display some non-local correlations. The user received these boxes from a provider, who claims that to produce the observed correlations the boxes perform some specific measurements on a given quantum state. The goal of the classical user is to make sure that the boxes work properly, i.e. that they contain the claimed state and perform the claimed measurements. This is especially relevant if the user does not trust the provider or, even if trusted, does not want to rely on the provider's ability to prepare the devices. Self-testing is the procedure that allows for this kind of certification. The self-tested states and measurements can later be used to run a given quantum information protocol, as proposed in~\cite{MayersYao} for secure key distribution. For many protocols, however, passing through self-testing techniques is not necessary and in fact it is simpler and more efficient to run the protocol directly from the observed correlations, as for example in standard device-independent quantum key distribution protocols~\cite{DIQKD}. Yet, self-testing protocols constitute an important device-independent primitive as they certify the entire description of the quantum setup only from the observed statistics.

As mentioned, the concept of self-testing was introduced by Mayers and Yao in~\cite{MayersYao}, where the procedure to self-test a maximally entangled pair of qubits is described. This protocol was made robust in subsequent works, see~\cite{MKYS,Magniez}. In the following years new self-testing protocols for more complicated states such as graph states were described \cite{McKague}, as well as protocols for self-testing more complicated operations, such as entire quantum circuits \cite{Magniez}. A general numerical method for self-testing, known as the SWAP method, was introduced in \cite{YVBSN}, providing much better estimations of robustness than the analytical proofs. This numerical method can also be used to self-test three-qubit states such as GHZ states~\cite{PVN} and W states~\cite{Wu}.

Despite its importance, we lack general techniques to construct and prove self-testing protocols. Most of the existing examples are built from the maximal violation of a Bell inequality. Based on geometrical considerations, see for instance~\cite{Werner,DPA}, one expects that generically there is a unique way, state and measurements, of producing the extremal correlations attaining the maximal quantum violation of a Bell inequality. This is not always the case, but whenever it is, we say that the corresponding Bell inequality is \textit{useful for self-testing}. Following this approach, it is possible to prove that the state and measurements maximally violating the Clauser-Horne-Shimony-Holt (CHSH) inequality \cite{CHSH} are unique~\cite{CHSHUniqueness,MKYS}, and the corresponding state is a maximally entangled two-qubit state. More recently, a self-testing protocol for any two-qubit entangled states has been derived in~\cite{Cedric} using the Bell inequalities introduced in~\cite{AMP}, and all the self-testing configurations for a maximally entangled state of two qubits using two measurements of two outputs have been identified in~\cite{Singapore}. From a general perspective, it is an interesting question to understand which Bell inequalities are useful for self-testing and what are the states and measurements certified by them. But, as seen in the previous discussion, little is known beyond the simple scenario involving two measurements of two outputs.

The main result of this work is to prove that the so-called chained Bell inequalities~\cite{chained}, defined for two devices performing an arbitrary number of measurements of two outputs, are useful for self-testing. Recall that the maximal violation of these inequalities is given by a maximally entangled two-qubit state and measurements equally spaced on an equator of the Bloch sphere~\cite{Wehner}. Our results imply that this known violation is unique. Our proof is based on a sum-of-squares (SOS) decomposition of the Bell operator defined by the chained  inequalities. The specific form of the SOS decomposition allows us to construct a quantum circuit that acts as a swap-gate, provided that the inequality is maximally violated. It is then proven that the swap-gate circuit correctly isolates the states and measurements that need to be self-tested, that is, those providing the maximal violation. 

\section{Preliminaries}

\subsection{Self-testing terminology}

 In this section we define the settings and introduce some self-testing terminology. We consider the standard Bell scenario in which two parties share a quantum state $\ket{\psi'}$ on which they can perform $n$ measurements, described by the two-outcome operators ${A_i', B_i'}$, where $i=1,\ldots,n$. The shared state and measurements are not trusted and are modelled as black boxes: each of them gets some classical input, which labels the choice of measurement, and produces a classical output, the measurement result. As the dimension is arbitrary, one can restrict the analysis to pure states and projective measurements without any loss of generality. The state $\ket{\psi'}$ lives in a joint Hilbert space $\mathcal{H}_A' \otimes \mathcal{H}_B'$ of an unknown dimension. Operators $A_i'(B_i')$ act on the part of the state living in $\mathcal{H}_A'(\mathcal{H}_B')$ , so that operators of different parties commute: $\left[A_i',B_i'\right] = 0$. Also, $M_{A_i'}^{\pm} = (\mathbbm{1} \pm A_i')/2$ and $M_{B_i'}^{\pm} = (\mathbbm{1} \pm B_i')/2$ can be considered to be projective measurements. In this scenario the parties calculate the joint outcome probabilities that can be described as $p(a,b|i,j) = \bra{\psi'} M_{A_i'}^{a} \otimes M_{B_i'}^b \ket{\psi'}$. The set of joint probabilities for all possible combinations of inputs and outputs is often simply called the set of correlations. The parties can also check whether the probability distribution is non-local, i.e. whether some Bell inequality is violated. A Bell inequality can be written as a linear combination of the observed correlations.

 Usually there is a specification of the black boxes, in self-testing terminology named as the \textit{reference experiment}, and it consists of the state $\ket{\psi} \in \mathcal{H}_A\otimes \mathcal{H}_B$ and measurements $A_i,B_i$ in some given Hilbert spaces $\mathcal{H}_A$ and $\mathcal{H}_B$ of finite dimension. On the other hand, the term \textit{physical experiment} is used for the actual state and measurements $\{\ket{\psi'}, A_i', B_i'\}$. The aim of self-testing is to compare the reference and the physical experiment and certify that they are physically equivalent. This means that the physical experiment is the same as the reference experiment up to local unitaries and additional non-relevant degrees of freedom, which are unavoidable, that is:
 \begin{eqnarray}
  \ket{\psi'}=U_{AA'}\otimes U_{BB'}\ket{\psi}_{AB}\ket{\varphi}_{A'B'} \nonumber\\
A'_i\otimes B'_i  \ket{\psi'}=U_{AA'}\otimes U_{BB'}\left(A_i\otimes B_i\ket{\psi}_{AB}\right)\ket{\varphi}_{A'B'} ,
 \end{eqnarray}
where %\replaced{
$\ket{\varphi_{A'B'}}$%}{$\ket{\varphi}_{A'}$ and $\ket{\varphi}_{B'}$}
 describe the local states of the possible additional degrees of freedom of the physical experiment and $U_{AA'}$ and $U_{BB'}$ are arbitrary local unitaries acting on systems $AA'$ and $BB'$. %\deleted{In principle, the state of the additional degrees of freedom does not have to be in a pure state, but for simplicity of notation and without loss of generality in this work we will assume purity}. 
 We introduce the product isometry $\Phi = \Phi_A \otimes \Phi_B :\mathcal{H}_A'\otimes \mathcal{H}_B' \rightarrow \mathcal{H}_A\otimes \mathcal{H}_B \otimes \mathcal{H}_A'\otimes \mathcal{H}_B'$, a map that preserves %\replaced{
 the inner product%}{distance}
 , but does not have to preserve dimension. Thus we say that a self-testing protocol is successful if there exists a local isometry relating the physical and reference experiment:
\begin{eqnarray}
\Phi\left(\ket{\psi'}\right) = \ket{\psi}\ket{\varphi} \nonumber\\
\Phi\left(A_i'\otimes B_i'\ket{\psi'}\right) = \left(A_i\otimes B_i \ket{\psi}\right)\ket{\varphi} .
\label{eq:equivalence}
\end{eqnarray}
In self-testing terminology the relation between the physical and the reference experiment described  by (\ref{eq:equivalence}) is called \textit{equivalence}. 

Trivially, a necessary condition for equivalence is that the full set of correlations obtained from the black boxes is equal to the set of correlations that one would obtain after applying the reference measurements on the reference state. A weaker necessary condition is to verify that the two sets of correlations lead to the same maximal quantum violation of a given Bell inequality.While in general checking all the correlations provides more information, there are some situations where observing just the maximum quantum value of a Bell inequality has been proven to be sufficient to certify the equivalence between the physical and the reference experiment. This is the approach we follow in this work and prove that the observation of the maximum quantum violation of the chained Bell inequalities suffices for self-testing.

\subsection{The chained Bell inequality}
\label{1}

The chained Bell inequalities were introduced in Refs. \cite{chained} to generalize the well-known Clauser-Horne-Shimony-Holt (CHSH) Bell inequality \cite{CHSH} to a larger number of measurements per party, while keeping the number of outcomes to two. Let us denote by $A_i$ and $B_i$ $(i=1,\ldots,n)$ the observables of Alice and Bob, respectively, and assume that they all have outcomes $\pm1$. Then, the chained Bell inequality for $n$ inputs reads
\begin{equation}
\label{eq:chainedBell}
\mathcal{I}_{\mathrm{ch}}^n := \sum_{i=1} ^{n}\left( \langle A_iB_i\rangle  + \langle A_{i+1} B_i\rangle \right)\leq 2n - 2.
\end{equation}
where we denote $A_{n+1}\equiv -A_1$. Notice that for $n=2$ the above formula reproduces
the CHSH Bell inequality
\begin{equation}
\langle A_1B_1\rangle+\langle A_2B_1\rangle+\langle A_2B_2\rangle-\langle A_1B_2\rangle\leq 2.
\end{equation}

Importantly, in quantum theory the chained Bell inequality can be violated by Alice and Bob if they perform measurements on an entangled quantum state. To be more precise, let there exist quantum observables $A_i$ and $B_i$, i.e., Hermitian operators with eigenvalues $\pm1$ acting on some Hilbert space $\mathcal{H}$ of, so far, unspecified dimension, and an entangled state $\ket{\psi}\in \mathcal{H}\otimes \mathcal{H}$ such that
$\langle\psi|\mathcal{B}_{n}|\psi\rangle>2n-2$, where $\mathcal{B}_{\mathrm{n}}$ stands for the
so-called \textit{Bell operator}
\begin{equation}
\mathcal{B}_{n} = \sum_{i=1} ^{n} \left(A_i\otimes B_i   + A_{i+1}\otimes  B_i\right),
\end{equation}
where again $A_{n+1}\equiv -A_1$. In particular, it has been shown in Ref. \cite{Wehner} that
the maximal quantum violation of the Bell inequality (\ref{eq:chainedBell}) amounts to
\begin{equation}
B_n^{\mathrm{\max}}=2n\cos\frac{\pi}{2n},
\end{equation}
and it is realized with the maximally entangled state of two qubits 
\begin{equation}\label{maxentstate}
\ket{\phi_+}=\frac{1}{\sqrt{2}}(\ket{00}+\ket{11}) ,
\end{equation}
and the following measurements
\begin{eqnarray}
\label{eq:ChainedMeasurementsA}
A_i = s_iX + c_iZ,\qquad B_i = s_i'X+c_i'Z,
\end{eqnarray}
where $X$ and $Z$ are the standard Pauli matrices and $s_i=\sin\phi_i$, $c_i=\cos\phi_i$, $s_i'=\sin\phi_i'$ and $c_i'=\cos\phi'_i$, where $\phi_i=[(i-1)\pi]/n$ and $\phi'_i=[(2i-1)\pi]/2n$ (see Fig. \ref{fig:measurements}).
\begin{figure}[h!]
\centering
\includegraphics[width=0.3\textwidth]{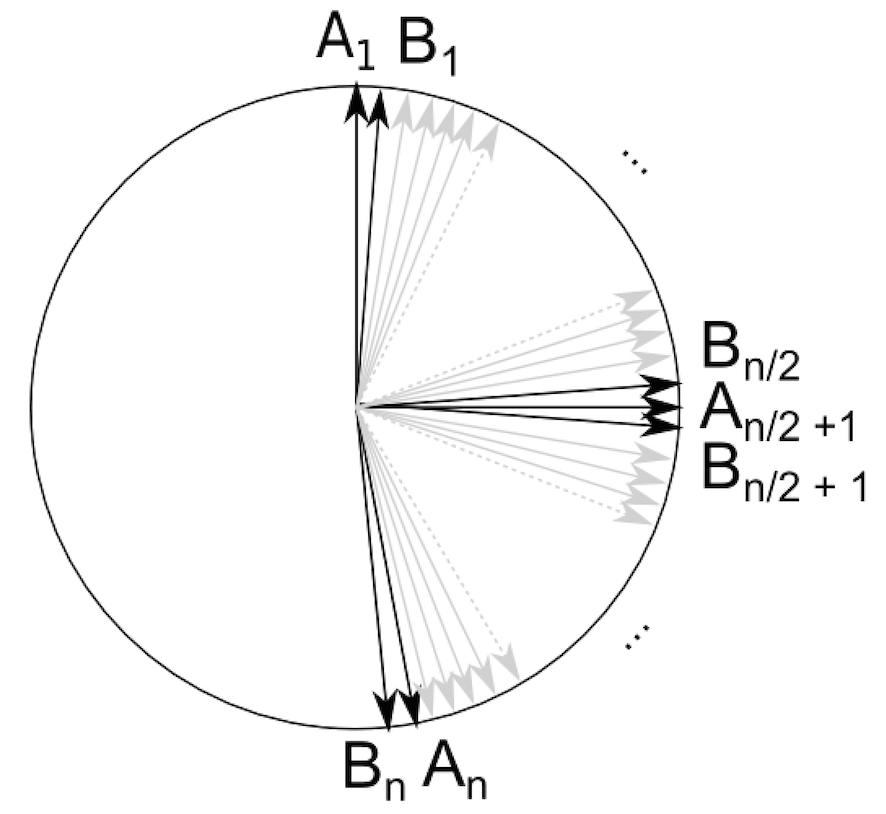}
\includegraphics[width=0.3\textwidth]{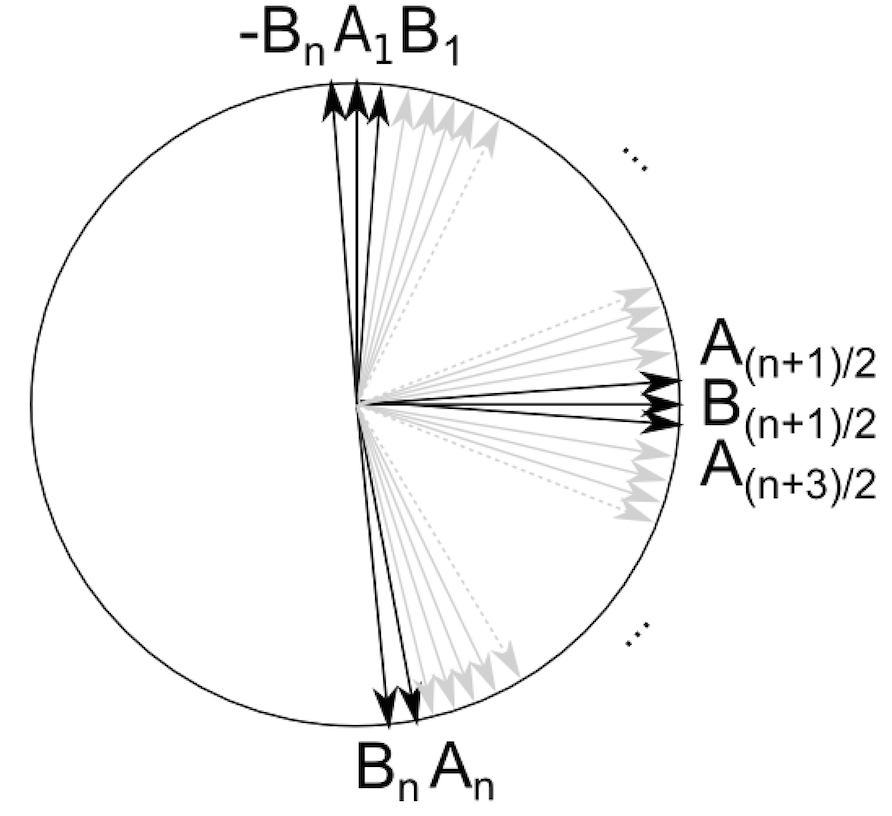}
\caption{The optimal measurements $A_i$ and $B_i$ depicted on the $XZ$ plane of the Bloch sphere with  $i -= 1,\ldots, n$. The case with an even number of measurements is on the left, and the odd case is on the right.\label{fig:measurements}}
\end{figure}

For further purposes, let us also introduce the notion of the \textit{shifted Bell operator}, that is, an operator
given by $B_n^{\max}\mathbbm{1}-\mathcal{B}_{\mathrm{n}}$. Since, by the very construction, this operator is positive-semidefinite, then there exist a finite number of operators $P_i$ (not necessarily positive) which are functions of the measurements $A_i$ and $B_i$ such that
\begin{equation}
\label{eq:BellSOS}
\mathcal{B}_s = B_n^{\max}\mathbbm{1} - \mathcal{B}_{\mathrm{n}} = \sum_i P_i^\dagger P_i.
\end{equation}
This decomposition is called a \textit{sum-of-squares} (SOS) decomposition, in this particular case of the shifted Bell operator. Furthermore, an SOS decomposition in which operators $P_i$ contain products of at most $n$ measurement operators is named \textit{SOS decomposition of n-th degree}. The use of SOS decompositions for self-testing proofs has been previously considered in~\cite{YN}. Numerically, it is possible to obtain SOS decompositions of various degrees via the Navascues-Pironio-Acin (NPA) hierarchy \cite{NPA}. In fact, the degree of the SOS decomposition is related to the level of the NPA hierarchy used. The dual of the semi-definite program defined by the {$n$-th} level of the NPA hierarchy yields an SOS decomposition of $n$-th degree.

What is important for further considerations is that if $\ket{\psi}$ maximally violates the chained Bell inequality,
then $(B_n^{\max}\mathbbm{1}-\mathcal{B}_{\mathrm{n}})\ket{\psi}=0$, which implies that $P_i\ket{\psi}=0$ for every $i$. In other words, $\ket{\psi}$ belongs to the intersection of kernels of the operators $P_i$. This imposes a plethora of conditions on the state and measurements maximally violating the Bell inequality.

\section{Self-testing with the chained Bell inequalities}
\label{2}

In this section we prove that with the aid of the chained Bell inequalities one can self-test the presence of the maximally entangled state (\ref{maxentstate}) and identify measurements (\ref{eq:ChainedMeasurementsA}), thus generalizing the results previously obtained for the CHSH Bell inequality in Refs. \cite{CHSHUniqueness,MKYS,MShi}. The advantage of our approach over the previous results lies on the fact that in the limit of a large number of measurements, the chained Bell inequality allows one to self-test the entire plane of the Bloch sphere spanned by the Pauli matrices $X$ and $Z$. Also, our results imply that the maximal quantum violation of the chained Bell inequalities is unique in the sense that there exists only one probability distribution maximally violating them. This makes chained Bell inequalities useful for randomness certification (see \cite{DPA}). In the context of nonlocal games this result confirms that measuring  (\ref{eq:ChainedMeasurementsA}) on a maximally entangled state state (\ref{maxentstate}) is the only way (up to local isometries) to win the odd cycle game with maximal probability; it is known that the probability to win the odd-cycle game in quantum regime is $\cos^2(\pi/4n)$ \cite{OC}.

\subsection{The SOS decompositions}

The key ingredient in our proof are the following two SOS decompositions of the shifted 
Bell operator associated to the chained Bell inequality whose proofs are deferred to 
\ref{app:SOS}. We start from the first degree SOS decomposition.

\begin{lem}\label{Lem1storder}
Let $\{\ket{\psi'},A_i',B_i'\}$ be the state and the measurements maximally violating the chained Bell inequality. 
Then, the corresponding shifted Bell operators admits the following SOS of first degree:
\begin{eqnarray}
\label{eq:genSOS1st}
B_n^{\max}\mathbbm{1} - \mathcal{B}_{n} &=&  \cos{\frac{\pi}{2n}}\left[ \sum_{ i = 1}^{n} \left( \mathbbm{1} - A_{i} \otimes  \frac{B_{i} + B_{i - 1}}{2 \cos{(\pi/2n)}} \right)^{2}\right. \nonumber\\
&&\left.\hspace{1.5cm} + \frac{1}{n}\sum_{j=1}^{n}\sum_{i = 1}^{n - 2} \left( \alpha_{i} B_{j} + \beta_{i} B_{i + j} + \gamma_{i} B_{ i + j + 1} \right)^{2} \right],
\end{eqnarray}
where we assume that $B_{n+j} = - B_{j}$ and $B_{n} = - B_{0}$. The coefficients $\alpha_{i}$, $\beta_{i}$, and $\gamma_{i}$ are given by
\begin{eqnarray}
\label{eq:coefficientsa}
\alpha_{i} = \frac{\sin{(\pi/n)}}{2\cos{(\pi/2n)}} \sqrt{\frac{1}{\sin{(\pi i/n)} \sin{ [\pi (i+1)/n]}}}, \\
\label{eq:coefficientsb}
\beta_{i} = \frac{-1}{2\cos{(\pi/2n)}}  \sqrt{\frac{\sin{[\pi (i+1)/n]}}{\sin{(\pi i/n)}}}, 
\end{eqnarray}
and
\begin{equation}\label{eq:coefficientsg}
\gamma_i=\frac{1}{2\cos{(\pi/2n)}}  \sqrt{\frac{\sin{(\pi i/n)}}{\sin{[\pi (i+1)/n]}}}=-\frac{1}{4\beta_i\cos^2(\pi/2n)}
%\gamma_{i} = \frac{1}{2\cos{\frac{\pi}{2n}}} \sqrt{\frac{\sin{\frac{\pi i}{n}}}{ \sin{\frac{\pi (i+1)}{n}}}}
\end{equation}
with $i = 1, \ldots, n - 2$.
\end{lem}

Note that the above SOS decomposition remains valid if in its second line we omit the sum over $j$ and fix $j$ to be any number from $\{1,\ldots,n\}$. Also, the transformations $A_i \rightarrow B_i$ and $B_i \rightarrow A_{i+1}$ in the first parenthesis, and $B_i \rightarrow A_i$ in the second one lead to the whole family of  $2n$ SOS decompositions. %\added{
Let us finally mention that that 
the above SOS decomposition is a particular case of an SOS decomposition for a more general Bell inequality which will be presented in Ref. \cite{SATPWA}
together with an analytical method used to derive it.%}  

It turns out, however, that none of them is enough for self-testing. In fact, we need the following second degree SOS decomposition.

\begin{lem}Let $\{\ket{\psi'},A_i',B_i'\}$ be the state and the measurements maximally violating the chained Bell inequality. Then, the corresponding shifted Bell operator admits the following second-order SOS:
\begin{eqnarray}
\label{eq:SOSnBell2ndorder}
\fl B_n^{\max}\mathbbm{1} - \mathcal{B}_n &\hspace{-0.5cm}=&
\frac{1}{8n\cos\case{\pi}{2n}}\bigg\{2(B^{\max}_n\mathbbm{1}-\mathcal{B}_n)^2
+\sum_{i,j = 1\atop j \neq i,i-1}^{n} \left[A_i \otimes (B_i + B_{i-1}) - (A_j + A_{j+1})\otimes B_{j}\right]^2 \nonumber\\ 
&&\hspace{2cm}+\sum_{i = 1}^{n}\left[\left(A_i \otimes B_i - A_{i+1} \otimes B_{i+1}\right)^2 +\left(A_i \otimes B_{i-1} - A_{i+1}\otimes B_{i}\right)^2\right]\bigg\} \nonumber\\ 
&&+ \frac{1}{2}\cos\left(\frac{\pi}{2n}\right) \sum_{i = 1}^{n-2}\left[\left( \alpha_{i} B_{1} + \beta_{i} B_{i + 1} + \gamma_{i} B_{ i + 2} \right)^2 + \left( \alpha_{i} A_{1} + \beta_{i} A_{i + 1} + \gamma_{i} A_{ i + 2} \right)^2\right],\nonumber\\
\end{eqnarray}
where we again used the notation $A_{n+1} = -A_{1}$ and $A_{0} = -A_{n}$ and the same for $B$'s, and the $\alpha_i$, $\beta_i$ and $\gamma_i$ are given in Eqs. (\ref{eq:coefficientsa}), (\ref{eq:coefficientsb}) and (\ref{eq:coefficientsg}).
\end{lem}
Similarly we can construct another SOS decomposition from the above one by applying the following transformations to it: $A_i \rightarrow B_i$ in all terms, $B_i \rightarrow A_{i+1}$ in the curly brackets and $B_i \rightarrow A_{i}$ in the remaining terms.

%Now we will use the established SOS decompositions to derive some relations for the operators from Eqs. %(\ref{eq:Remigiusz1}) and (\ref{eq:Remigiusz2}). 

\subsection{Exact case}

We start our considerations with the ideal case when the black boxes reach the maximal quantum violation of the Bell inequality and and leave the study of the robustness of our schemes for the following section. 

\begin{figure}[h!]
\centering
\includegraphics[width=0.6\textwidth]{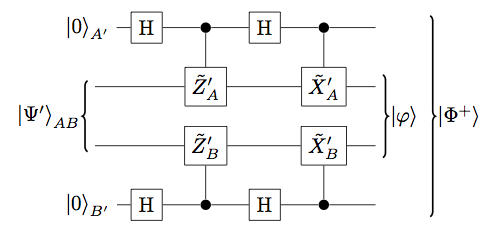}
\caption{The swap-gate used for self-testing. In it, $\ket{\psi'}_{AB}$ stands for 
the state maximally violating the given Bell inequality, while $\ket{0}_{A'}$ and $\ket{0}_{B'}$ are ancillary qubit states controlling the gates $\tilde{X}_A$,$\tilde{Z}_A$,$\tilde{X}_B$ and $\tilde{Z}_B$. Then, $H$ is the standard one-qubit Hadamard gate defined in the text. $\tilde{X}_A$,$\tilde{Z}_A$,$\tilde{X}_B$ and $\tilde{Z}_B$ are regularized, if necessary, versions of $X_A'$,$Z_A'$,$X_B'$ and $Z_B'$ respectively. At the output of the circuit the ancillary qubits are in the desired state $\ket{\phi_+}$.\label{fig:swapgate}}
\end{figure}

The departure point of our considerations is the swap-gate introduced in Ref. \cite{Mosca} and presented on Fig. \ref{fig:swapgate}. In what follows we show that with properly chosen controlled gates $X_A'$, $Z_B'$, $X_B'$ and $Z_B'$ it defines a unitary operation that satisfies Eq. (\ref{eq:equivalence}). To this end, let us choose 
\begin{equation}
\label{eq:Remigiusz1}
X_A'=\left\{
\begin{array}{ll}
A_{n/2+1}', & n\,\,\mathrm{even}\\[2ex]
\displaystyle\frac{A_{(n+1)/2}'+A_{(n+3)/2}'}{2\cos(\pi/2n)}, & n\,\,\mathrm{odd}
\end{array}\right.,\qquad Z_A'=A_1
\end{equation}
and
\begin{equation}
\label{eq:Remigiusz2}
X_B'=\left\{
\begin{array}{ll}
\displaystyle\frac{B_{n/2}'+B_{n/2+1}'}{2\cos(\pi/2n)}, & n\,\,\mathrm{even}\\[2ex]
B_{(n+1)/2}, & n\,\,\mathrm{odd}
\end{array}\right.,\qquad Z_B'=\frac{B_1'-B_n'}{2\cos(\pi/2n)}.
\end{equation}
Clearly, as all observables $A_i'$ and $B_i'$ are Hermitian and have eigenvalues $\pm1$, $Z_A'$ and $X_A'$ for even $n$ and $X_B'$ for odd $n$ are unitary. However, the operators $X_A'$ for odd $n$, $X_B'$ for even $n$ and $Z_B'$ might not be unitary in general, which in turn makes the circuit of Figure \ref{fig:swapgate} non-unitary. To overcome this problem we exploit the polar decomposition which says that one can write any operator $M$ as $M=U|M|=|M|V$ where $U$ and $V$ are some unitary operators and $|M|=\sqrt{M^{\dagger}M}$. Then, if $X_B'$ and $Z_B'$ are of full rank we define $\widetilde{X}_B=X'_B/|X'_B|$ and $\widetilde{Z}_B=Z'_B/|Z'_B|$, while if one of them is rank deficient, say $Z_B'$, we replace its zero eigenvalues by one and then use the above construction; in other words, we define $\widetilde{Z}_B=(Z_B'+P)/|Z_B'+P|$ with $P$ denoting the projector onto the kernel of $Z'_B$.

First, notice that it follows from the SOS decompositions (\ref{eq:genSOS1st}) and (\ref{eq:SOSnBell2ndorder})
that for any $i=1,\ldots,n$, the identities
\begin{equation}
A_i'\otimes \frac{B'_i+B_{i-1}'}{2\cos(\pi/2n)}\ket{\psi'}=\ket{\psi'},\qquad 
\frac{A'_i+A'_{i+1}}{2\cos(\pi/2n)}\otimes B_i\ket{\psi'}=\ket{\psi'}
\end{equation}
are satisfied, which imply in particular that 
\begin{equation}\label{properties}
X_A'\ket{\psi'}=X_B'\ket{\psi'},\qquad Z_A'\ket{\psi'}=Z_B'\ket{\psi'}.
\end{equation}
Moreover, one can prove that (see \ref{app:exact}) the operators $X'_A$ and $Z'_A$ 
anticommute in the following sense
\begin{equation}\label{AntiExact}
\{X'_A,Z'_A\}\ket{\psi'}=0.
\end{equation}
Finally, although the tilded operators are in general different than $X'_B$ and $Z'_B$, it turns out that 
they act in the same way when applied to $\ket{\psi'}$, that is,
\begin{equation}\label{properties2}
\widetilde{X}_B\ket{\psi'}=X'_B\ket{\psi'},\qquad \widetilde{Z}_B\ket{\psi'}=Z'_B\ket{\psi'}.
\end{equation}
To prove these relations, let $\|\cdot\|$ stand for the vector norm defined as $\|\ket{\psi}\|=\sqrt{\langle\psi|\psi\rangle}$. Then, the following reasoning applies \cite{Cedric}
\begin{eqnarray}\label{Faustino}
\|(\widetilde{X}_B-X'_B)\ket{\psi'}\|&=&
%\|\widetilde{X}_B(\mathbbm{1}-\widetilde{X}^{\dagger}_BX_B)\ket{\psi'}\|
\|(\mathbbm{1}-\widetilde{X}^{\dagger}_BX'_B)\ket{\psi'}\|
=\|(\mathbbm{1}-|X'_B|)\ket{\psi'}\|\nonumber\\
&=&\|(\mathbbm{1}-|X'_AX'_B|)\ket{\psi'}\|\leq \|(\mathbbm{1}-X'_AX'_B)\ket{\psi'}\|=0,
\end{eqnarray}
where the first and the second equalities stem from the fact that $\widetilde{X}_B$ is unitary and its definition, respectively. The third equality is a consequence of the fact $X_A'$ is unitary which implies that $|X'_AX'_B|=|X_B'|$, and, finally, the inequality and the last equality follow from the operator inequality $M\leq |M|$ and Eq. (\ref{properties}).

We are now ready to state and prove our first main result.

\begin{thm}\label{thm1}
Let $\{\ket{\psi'}, A_i', B_i'\}$ be the state and the measurements maximally violating the chained Bell inequality (\ref{eq:chainedBell}). Then the unitary operation $\Phi$ defined above is such that for any pair $i,j=1,\ldots,n$
\begin{eqnarray}\label{thm1:0}
& \Phi(A_i'B_j'\ket{\psi'}\ket{00})=\ket{\varphi}A_iB_j\ket{\phi_+},\\
\label{thm1:1}
& \Phi(A_i'\ket{\psi'}\ket{00})=\ket{\varphi}A_i\ket{\phi_+}, & \Phi(B_j'\ket{\psi'}\ket{00})=\ket{\varphi}B_j\ket{\phi_+},\\
\label{thm1:2}
& \Phi(\ket{\psi'}\ket{00})=\ket{\varphi}\ket{\phi_+},
\end{eqnarray}
where $\ket{\varphi}$ is some state, $\ket{\phi_+}$ is the two-qubit maximally entangled state, and $A_i$ and $B_i$ are given by Eq. (\ref{eq:ChainedMeasurementsA}).%, and, finally, it is assumed that $A_0'=B_0'=\mathbbm{1}$.
\end{thm}
\begin{proof}
Let us first consider Eq. (\ref{thm1:0}). Owing to the linearity of $\Phi$ in both Alice's and Bob's measurements and the fact that for even $n$ (see Lemma \ref{Structure} in \ref{app:exact}):
\begin{eqnarray}
\label{eq:ABstructuremain1}
A_i'\ket{\psi'} = \left(s_iX_A' + c_iZ_A'\right)\ket{\psi'},\qquad 
\label{eq:ABstructuremain2}
B_i'\ket{\psi'} = \left(s_i'X_B' + c_i'Z_B'\right)\ket{\psi'},
\end{eqnarray}
the left-hand side of Eq. (\ref{thm1:0}) can be rewritten as
\begin{eqnarray}\label{Rioja}
\Phi(A_i'B_j'\ket{\psi'}\ket{00})&=&s_is_i'\Phi(X_A'X_B'\ket{\psi'}\ket{00})+s_ic_i'\Phi(X_A'Z_B'\ket{\psi'}\ket{00})\nonumber\\
&&+c_is_i'\Phi(Z_A'X_B'\ket{\psi'}\ket{00})+c_ic_i'\Phi(Z_A'Z_B'\ket{\psi'}\ket{00}).
\end{eqnarray}
Then, it follows from Eqs. (\ref{properties}) and (\ref{AntiExact}) that $X_A'X_B'\ket{\psi'}=Z_A'Z_B'\ket{\psi'}=\ket{\psi'}$ and $X_A'Z_B'\ket{\psi'}=-Z_A'X_B'\ket{\psi'}$, and therefore we only need to check how the map $\Phi$ applies to $\ket{\psi'}$ and $X_A'Z_B'\ket{\psi'}$. In the first case, one has
\begin{eqnarray}
\label{eq:isometry}
\fl\Phi(\ket{\psi'}\ket{00})= \frac{1}{4}\big[(\mathbbm{1} + Z_A')(\mathbbm{1} + \widetilde{Z}_B)\ket{\psi'}\ket{00}
+ X_A'(\mathbbm{1} - Z_A')(\mathbbm{1} + \widetilde{Z}_B)\ket{\psi'}\ket{10}\nonumber\\
\hspace{0.5cm}+ \widetilde{X}_B(\mathbbm{1} + Z_A')(\mathbbm{1} - \widetilde{Z}_B)\ket{\psi'}\ket{01}+ X_A'\widetilde{X}_B(\mathbbm{1} - Z_A')(\mathbbm{1} - \widetilde{Z}_B)\ket{\psi'}\ket{11}\big].
\end{eqnarray} 
Exploiting Eqs. (\ref{properties}) and (\ref{properties2}) to convert $\widetilde{Z}_B$ to $Z_B'$ and then $Z_B'$ 
to $Z_A'$, and the fact that 
$Z_A'$ has eigenvalues $\pm1$, meaning that $(\mathbbm{1}+Z_A')$ and $(\mathbbm{1}-Z_A')$ are projectors onto orthogonal subspaces, one finds that the terms in Eq. (\ref{eq:isometry}) containing the ancillary vectors 
$\ket{01}$ and $\ket{10}$ simply vanish, and the whole expression simplifies to
\begin{eqnarray}
\label{eq:isometry2}
\Phi(\ket{\psi'}\ket{00})= \frac{1}{4}\big[(\mathbbm{1} + Z_A')^2\ket{\psi'}\ket{00}
+ X_A'\widetilde{X}_B(\mathbbm{1} - Z_A')^2\ket{\psi'}\ket{11}\big].
\end{eqnarray} 
Using then the fact that $(\mathbbm{1}\pm Z'_A)^2=2(\mathbbm{1}\pm Z'_A)$, 
the anticommutation relation (\ref{AntiExact}) and the identities (\ref{properties}) and (\ref{properties2}),
we finally obtain 
\begin{equation}\label{Key1}
\Phi(\ket{\psi'}\ket{00})=\ket{\varphi}\ket{\phi_+}
\end{equation}
with 
$\ket{\varphi}=(1/2\sqrt{2})(\mathbbm{1}+Z_A')^2\ket{\psi'}$, which is exactly Eq. (\ref{thm1:2}). %for $i=j=0$.

In the second case, i.e., that of $\Phi(X_A'Z_B'\ket{\psi'}\ket{00})$, one has  
\begin{eqnarray}
\label{eq:isometry3}
\fl\Phi(X'_AZ'_B\ket{\psi'}\ket{00})&=& \frac{1}{4}\big[(\mathbbm{1} + Z'_A)(\mathbbm{1} + \widetilde{Z}_B)X'_AZ'_B\ket{\psi'}\ket{00}+ X'_A(\mathbbm{1} - Z'_A)(\mathbbm{1} + \widetilde{Z}_B)X'_AZ'_B\ket{\psi'}\ket{10}\nonumber\\
&&\hspace{0.5cm}+ \widetilde{X}_B(\mathbbm{1} + Z'_A)(\mathbbm{1} - \widetilde{Z}_B)X'_AZ'_B\ket{\psi'}\ket{01}\nonumber\\
&&\hspace{0.5cm}+ X'_A\widetilde{X}_B(\mathbbm{1} - Z'_A)(\mathbbm{1} - \widetilde{Z}_B)X'_AZ'_B\ket{\psi'}\ket{11}\big].
\end{eqnarray} 
Exploiting the properties (\ref{properties}) and (\ref{properties2}), 
the anticommutation relation (\ref{AntiExact}), and the fact that  
$(\mathbbm{1}+Z_A')(\mathbbm{1}-Z'_A)=0$, one can prove that the terms in Eq. (\ref{eq:isometry3}) containing kets $\ket{00}$ and $\ket{11}$ are zero and the whole expression reduces to 
\begin{eqnarray}
\label{eq:isometry4}
\Phi(X'_AZ'_B\ket{\psi'}\ket{00})&=& \frac{1}{4}\left[(\mathbbm{1} + Z'_A)^2\ket{\psi'}\ket{10}+ X'_AZ'_A\widetilde{X}_B(\mathbbm{1} - Z'_A)^2\ket{\psi'}\ket{01}\right].
\end{eqnarray} 
By applying then Eq. (\ref{properties}) and the anticommutation relation (\ref{AntiExact}) in the second term of Eq. (\ref{eq:isometry4}), one can 
rewrite it as 
\begin{equation}\label{Key2}
\Phi(X'_AZ'_B\ket{\psi'}\ket{00})=\ket{\varphi}X_AZ_B\ket{\phi_+}.
\end{equation}
After plugging Eqs. (\ref{Key1}) and (\ref{Key2}) into Eq. (\ref{Rioja})
and using the fact that the Pauli matrices $X$ and $Z$ anticommute and satisfy $X_AX_B\ket{\phi_+}=Z_AZ_B\ket{\phi_+}=\ket{\phi_+}$, we arrive at 
\begin{eqnarray}\label{Rioja2}
\Phi(A_i'B_j'\ket{\psi'}\ket{00})&=&s_is_i'\ket{\varphi}X_AX_B\ket{\phi_+}+s_ic_i'\ket{\varphi}X_AZ_B\ket{\phi_+}\nonumber\\
&&+c_is_i'\ket{\varphi}Z_AX_B\ket{\phi_+}+c_ic_i'\ket{\varphi}Z_AZ_B\ket{\phi_+},
\end{eqnarray}
which by virtue of the formulas (\ref{eq:ChainedMeasurementsA}) is exactly Eq. (\ref{thm1:0}). % with 
%$i,j>0$.

Let us now prove Eqs. (\ref{thm1:1}). %consider the case when either $i$ or $j$ is zero.
From the the linearity of $\Phi$ and Eq. (\ref{eq:ABstructuremain1}), we get
\begin{equation}\label{Codorniu}
\Phi(A_i'\ket{\psi'}\ket{00})=s_i\Phi(X_A'\ket{\psi'}\ket{00})+c_i\Phi(Z_A'\ket{\psi'}\ket{00}).\nonumber\\
\end{equation}
Following the same steps as above, one can prove the following relations 
\begin{equation}
\Phi(X_A'\ket{\psi'}\ket{00})=\ket{\varphi}X_A\ket{\phi_+}, \qquad
\Phi(Z_A'\ket{\psi'}\ket{00})=\ket{\varphi}Z_A\ket{\phi_+},
\end{equation}
which when plugged into Eq. (\ref{Codorniu}) leads, in virtue of Eq. (\ref{eq:ABstructuremain1}), to the first part of Eq. (\ref{thm1:1}).The second part of the same equation can be proven in exactly the same way.
\end{proof}
%
%\added{
\begin{cor}
An important corollary following directly from Theorem \ref{thm1}
is that the probability distribution $\{p(a,b|i,j)\}$ with
\begin{equation}
p(a,b|i,j)=\langle\psi'|M^a_{A_i'}\otimes M^b_{B_j'}|\psi'\rangle
\end{equation}
being the conditional probability of obtaining the outcomes $a$ and $b$ upon performing the $i$th and $j$th measurement, respectively, is unique. In other words, there is no other probability distribution maximally violating inequality (\ref{eq:chainedBell}) different than the one above. 
\end{cor}%}
Let us %\added{
also%}
 notice that in order to prove the uniqueness of correlations maximally violating the chained Bell inequality one needs only the conditions (\ref{thm1:1}) and (\ref{thm1:2}); the conditions
(\ref{thm1:0}) are superfluous. This is because 
\begin{eqnarray}
\langle\psi'|A_i'\otimes B_j'|\psi'\rangle
&=&(\langle 00|\langle\psi'|A_i')\Phi^{\dagger}
\Phi(B_j'\ket{\psi'}\ket{00})\nonumber\\
&=&\langle\phi_+| A_i\otimes B_j|\phi_+\rangle,
\end{eqnarray}
where the first equality follows from the fact that $\Phi$
is unitary and and second from Eqs. (\ref{thm1:1}) and (\ref{thm1:2}).

\section{Robustness}
\label{3}
For practical purposes, it is important to estimate the robustness of self-testing procedures, as in any realistic situation it is impossible due to experimental imperfections to actually reach the maximal violation of any Bell inequality. One expects, however, self-testing procedures to tolerate some deviations from the ideal case, that is, if the violation of the given Bell inequality is close to its maximum quantum value, the state producing the violation must be close to the state maximally violating this Bell inequality. In \cite{YN} it has been proven that SOS decompositions allow one to reach the best known robustness of all analytical self-test protocols.

Here we study how robust is the above self-testing procedure based on the chained Bell inequality. Assuming that the physical state $\ket{\psi'}$ and the physical measurements $A_i'$ and $B_i'$ violate the chained Bell inequality by
$B_n^{\max}-\varepsilon$ with some sufficiently small $\varepsilon>0$, we estimate 
the distance between $\ket{\psi'}$ and the reference state, and how this distance is affected when 
physical measurements are applied to it. For simplicity and clearness we give bounds for the case when the number of measurements is even; the bounds for in the odd $n$ case can be determined in an analogous way.

Let us begin by noticing that now $\langle\psi'|(B_{n}^{\max}\mathbbm{1}-\mathcal{B}_{\mathrm{n}})|\psi'\rangle=\varepsilon$, and therefore the exact relations (\ref{nierownosc1}), (\ref{nierownosc2}) and (\ref{nierownosc3}) do not hold anymore. We then need to derive their approximate versions. First, 
it stems from the first SOS decomposition that (see Lemma \ref{lemmaRobust} in \ref{app:rob})
\begin{equation}\label{nierownosc1}
\|(X_A'-X_B')\ket{\psi'}\|\leq \sqrt{\varepsilon_1(n)},\qquad
\|(Z_A'-Z_B')\ket{\psi'}\|\leq \sqrt{\varepsilon_1(n)},
\end{equation}
where $\varepsilon_1=\varepsilon/\cos(\pi/2n)$. Clearly, for any $n$, $\varepsilon_1(n)\leq \sqrt{2}$ and $\varepsilon_1(n)\to 0$ for $\varepsilon\to 0$. Moreover, following the same reasoning as in 
(\ref{Faustino}), one proves that 
\begin{equation}\label{nierownosc2}
\|(\widetilde{X}_B'-X_B')\ket{\psi'}\|\leq \sqrt{\varepsilon_1(n)},\qquad
\|(\widetilde{Z}_B'-Z_B')\ket{\psi'}\|\leq \sqrt{\varepsilon_1(n)}.
\end{equation}
Finally, both SOS decompositions (\ref{eq:genSOS1st}) and (\ref{eq:SOSnBell2ndorder}) imply the following 
approximate anticommutation relations (see Lemma \ref{Rob:lem1} in \ref{app:rob}):
\begin{equation}\label{nierownosc3}
\|\{X_A',Z_A'\}\ket{\psi'}\|\leq  \sqrt{2\varepsilon_1(n)}+\frac{1}{\xi_{n/2-1}}\left(\frac{4\sqrt{\epsilon_1(n)}}{\alpha_{n/2-1}}+n\sqrt{2\varepsilon_2(n)}\right) = \omega_{\mathrm{ev}}(n),
\end{equation}
where $\xi_i = 2\cos(2i+1)\pi/2n$, $\alpha_i$ is defined in Lemma \ref{Lem1storder}, and $\varepsilon_1$ and $\varepsilon_2$ are given in Lemma \ref{Rob:lem1} in \ref{app:rob}. In what follows we drop the dependence of $\varepsilon_1$ and $\varepsilon_2$ on $n$.

Equipped with these tools we can state and prove the second main result of this paper.

\begin{thm}\label{Rob:thm1}
Let $\{\ket{\psi'},A_i',B_i'\}$ be a state and measurements giving 
violation of the chained Bell inequality $B_n^{\max}-\varepsilon$. Then,
\begin{eqnarray}
\label{Rob:thm1:1}
& \|\Phi(A_i'B_j'\ket{\psi'}\ket{00})-\ket{\varphi}A_iB_j\ket{\phi_+}\|\leq f_{ij}(\varepsilon,n), \\
\label{Rob:thm1:2}
& \|\Phi(A_i'\ket{\psi'}\ket{00})-\ket{\varphi}A_i\ket{\phi_+}\|\leq f_{A_i}(\varepsilon,n), \\
\label{Rob:thm1:3} 
& \|\Phi(B_j'\ket{\psi'}\ket{00})-\ket{\varphi}B_j\ket{\phi_+}\|\leq f_{B_j}(\varepsilon,n),   \\
\label{Rob:thm1:4}
& \|\Phi(\ket{\psi'}\ket{00})-\ket{\varphi}\ket{\phi_+}\|\leq f(\varepsilon,n), 
\end{eqnarray}
where $i,j=1,\ldots,n$, $\Phi$ is the unitary transformation defined above, $\ket{\varphi}=(1/N)(\mathbbm{1}+Z_A')(\mathbbm{1}+\widetilde{Z}_B')\ket{\psi'}$ with $N$ denoting the length of $\ket{\varphi}$. The functions $f(\varepsilon,n)$, $f_{B_j}(\varepsilon,n)$, $f_{A_i}(\varepsilon,n)$ and $f_{ij}(\varepsilon,n)$ vanish as 
$\varepsilon\to 0$ and for sufficiently large $n$ scale with $n$ as $n^2$.
\end{thm}
\begin{proof}As the norm $N$ of $\ket{\varphi}$ cannot be computed exactly, it turns out that 
to prove this theorem it is more convenient to first estimate the following distance
\begin{equation}
\|\Phi(A_i'B_j'\ket{\psi'}\ket{00})-\ket{\varphi'}A_iB_j\ket{\phi_+}\|
\end{equation}
with
\begin{equation}
\label{eq:phi'}
\ket{\varphi'}=\frac{1}{2\sqrt{2}}(\mathbbm{1}+Z_A')(\mathbbm{1}+\widetilde{Z}_B')\ket{\psi'}.
\end{equation}
and then show that the error we have by doing so is small for sufficiently 
small $\varepsilon$. 

From now on we will mainly follow the steps of the proof of Theorem \ref{thm1} replacing 
the identities by the corresponding inequalities. First, let us notice that  
for any $i=1,\ldots,n$ (see \ref{app:rob} for the proof):
\begin{equation}\label{Fargo}
\|[A_i'-(s_iX_A'+c_iZ_A')]\ket{\psi'}\|\leq g_{\mathrm{ev}},\qquad
\|[B_i'-(s_i'X_B'+c_i'Z_B')]\ket{\psi'}\|\leq h_{\mathrm{ev}}.
\end{equation}
where $g_{\mathrm{ev}}$ and $h_{\mathrm{ev}}$ are given in Lemma \ref{lastlemma} of the Appendix.
Denoting by $\overline{A}_i$ and $\overline{B}_i$ the operators appearing in the parentheses in (\ref{Fargo}),
we can write
\begin{eqnarray}\label{Orange}
\fl\|\Phi(A_i'B_j'\ket{\psi'}\ket{00})-\ket{\varphi'}A_iB_j\ket{\phi_+}\|&\leq& \|\Phi(A_i'B_j'\ket{\psi'}\ket{00})-\Phi(\overline{A}_i\overline{B}_j\ket{\psi'}\ket{00})\|\nonumber\\
&&+\|\Phi(\overline{A}_i\overline{B}_j\ket{\psi'}\ket{00})-\ket{\varphi'}A_iB_j\ket{\phi_+}\|,
\end{eqnarray}
and, by further exploitation of the fact that $\Phi$ is unitary, the first norm can be upper bounded as
\begin{eqnarray}\label{Homeland}
\|\Phi(A_i'B_j'\ket{\psi'}\ket{00})-\Phi(\overline{A}_i\overline{B}_j\ket{\psi'}\ket{00})\|&\leq& 
\|(A_i'B_j'-\overline{A}_i\overline{B}_j)\ket{\psi'}\|\nonumber\\
&\leq &\|(A_i'-\overline{A}_i)\ket{\psi'}\|+\|(B_j'-\overline{B}_j)\ket{\psi'}\|\nonumber\\
&\leq & g_{\mathrm{ev}}+h_{\mathrm{ev}},
\end{eqnarray}
where to obtain the second inequality we have used the standard trick of 
adding and subtracting the term $A_i'\overline{B}_j\ket{\psi'}$, the triangle inequality for the norm, and 
the fact that $A_i$ is unitary and that the spectral radius of $\overline{B}_j$ is not larger than one.
The third inequality in (\ref{Homeland}) stems directly from (\ref{Fargo}). In the cases when $A_i'$ or $B_j'$
are equal to the identity operator $\mathbbm{1}$, the above bound is replaced by $h_{\mathrm{ev}}$ and $g_{\mathrm{ev}}$, respectively, while 
in the case $A_i'=B_j'=\mathbbm{1}$, this distance is simply zero.

Let us then concentrate on the second norm on the right-hand side of 
(\ref{Orange}). Exploiting the explicit forms of the operators $\overline{A}_i$ and $\overline{B}_i$ 
and the measurements $A_i$ and $B_i$,
one has
\begin{eqnarray}\label{czekolada}
\fl\|\Phi(\overline{A}_i\overline{B}_j\ket{\psi'}\ket{00})-\ket{\varphi'}A_iB_j\ket{\phi_+}\|&\leq&
\|\Phi(X_A'X_B'\ket{\psi'}\ket{00})-\ket{\varphi'}X_AX_B\ket{\phi_+}\|\nonumber\\
&&+\|\Phi(X_A'Z_B'\ket{\psi'}\ket{00})-\ket{\varphi'}X_AZ_B\ket{\phi_+}\|\nonumber\\
&&+\|\Phi(Z_A'X_B'\ket{\psi'}\ket{00})-\ket{\varphi'}Z_AX_B\ket{\phi_+}\|\nonumber\\
&&+\|\Phi(Z_A'Z_B'\ket{\psi'}\ket{00})-\ket{\varphi'}Z_AZ_B\ket{\phi_+}\|.
\end{eqnarray}
Let us consider the first and the last norm on the right-hand side of this inequality. 
With the aid of inequalities (\ref{nierownosc1}) and the fact that $X_AX_B\ket{\phi_+}=Z_AZ_B\ket{\phi_+}=\ket{\phi_+}$ 
both can be upper bounded by $\sqrt{\varepsilon_1}+\|\Phi(\ket{\psi'}\ket{00})-\ket{\varphi'}\ket{\phi_+}\|$.
Then, from the definition of the unitary operation $\Phi$ and the state $\ket{\varphi'}$ it follows that the latter norm can be upper bounded as
\begin{eqnarray}\label{Qaqi}
\fl\|\Phi(\ket{\psi'}\ket{00})-\ket{\varphi'}\ket{\phi_+}\|&\leq&\frac{1}{4}\left(\|X_A(\mathbbm{1}-Z_A)(\mathbbm{1}+\widetilde{Z}_B)\ket{\psi'}\|+
\|\widetilde{X}_B(\mathbbm{1}+Z_A)(\mathbbm{1}-\widetilde{Z}_B)\ket{\psi'}\|\right.\nonumber\\
&&\hspace{0.5cm}+\left.\|X_A\widetilde{X}_B(\mathbbm{1}-Z_A)(\mathbbm{1}-\widetilde{Z}_B)\ket{\psi'}-\ket{\varphi'}\|\right).
\end{eqnarray}
To upper bound the first two norms in (\ref{Qaqi}) we first exploit inequalities (\ref{nierownosc1}) and (\ref{nierownosc2})
which allow us to ``convert'' $\widetilde{Z}_B$ to $Z_B$ and then $Z_B$ to $Z_A$ introducing an error of $8\sqrt{\varepsilon_1}$, and
then we use the fact that $(\mathbbm{1}+Z_A')(\mathbbm{1}-Z_A')=0$. To upper bound the last norm in 
(\ref{Qaqi}) we first use the anticommutation relation (\ref{nierownosc3}) which leads us to 
\begin{equation}
\fl\|X_A\widetilde{X}_B(\mathbbm{1}-Z_A)(\mathbbm{1}-\widetilde{Z}_B)\ket{\psi'}-\ket{\varphi'}\|\leq 2\omega_{\mathrm{ev}}(n)+
2\|X_A\widetilde{X}_B(\mathbbm{1}-\widetilde{Z}_B)\ket{\psi'}-(\mathbbm{1}+\widetilde{Z}_B)\ket{\psi'}\|.
\end{equation} 
One then uses again inequalities (\ref{nierownosc1}) and (\ref{nierownosc2}) in order to ``convert'' 
$\widetilde{Z}_B$ to $Z_B$ and then $Z_B$ to $Z_A$. This gives
\begin{eqnarray}
\fl\|X_A\widetilde{X}_B(\mathbbm{1}-Z_A)(\mathbbm{1}-\widetilde{Z}_B)\ket{\psi'}-\ket{\varphi'}\|&\leq& 2\omega_{\mathrm{ev}}(n)+8\sqrt{\varepsilon_1}\nonumber\\
&&+2\|X_A\widetilde{X}_B(\mathbbm{1}-Z_A)\ket{\psi'}-(\mathbbm{1}+Z_A)\ket{\psi'}\|.
\end{eqnarray} 
After applying (\ref{nierownosc3}) and then (\ref{nierownosc1}) and (\ref{nierownosc2}), one finally arrives at
\begin{equation}
\|X_A\widetilde{X}_B(\mathbbm{1}-Z_A)(\mathbbm{1}-\widetilde{Z}_B)\ket{\psi'}-\ket{\varphi'}\|\leq 
4\omega_{\mathrm{ev}}(n)+16\sqrt{\varepsilon_1}.
\end{equation}
Taking all this into account, we have that 
\begin{equation}\label{GranSangreDeToro}
\|\Phi(\ket{\psi'}\ket{00})-\ket{\varphi'}\ket{\phi_+}\|\leq 6\sqrt{\varepsilon_1}+\omega_{\mathrm{ev}}(n).
\end{equation}

Let us now pass to the second norm in (\ref{czekolada}) and notice that by using 
inequality (\ref{nierownosc1}) and the fact that $Z_B\ket{\phi_+}=Z_A\ket{\phi_+}$, it can be upper bounded 
in the following way
\begin{eqnarray}\label{Vader}
\fl\|\Phi(X_A'Z_B'\ket{\psi'}\ket{00})-\ket{\varphi'}X_AZ_B\ket{\phi_+}\|&\leq &\sqrt{\varepsilon_1}+\frac{1}{4}
\left(\|(\mathbbm{1}+Z_A')(\mathbbm{1}+\widetilde{Z}_B)X_A'Z_A'\ket{\psi'}\|\right.\nonumber\\
&&\left.+\|X_A'\widetilde{X}_B(\mathbbm{1}+Z_A')(\mathbbm{1}+\widetilde{Z}_B)X_A'Z_A'\ket{\psi'}\|\right)\nonumber\\
&&+\|X_A'(\mathbbm{1}-Z_A')(\mathbbm{1}+\widetilde{Z}_B)X_A'Z_A'\ket{\psi'}-\ket{\varphi'}\|\nonumber\\
&&\left.+\|\widetilde{X}_B(\mathbbm{1}+Z_A')(\mathbbm{1}-\widetilde{Z}_B)X_A'Z_A'\ket{\psi'}+\ket{\varphi'}\|\right).
\end{eqnarray}
Let us consider the first two norms appearing on the right-hand side of (\ref{Vader}). 
Exploiting the anticommutation relation (\ref{nierownosc3}) and then inequalities 
(\ref{nierownosc1}) and (\ref{nierownosc2}) to convert $\widetilde{Z}_B$ to $Z_A$, we can bound each of 
these norms by $4\sqrt{\varepsilon_1}+2\omega_{\mathrm{ev}}(n)$. Using then the inequality (\ref{nierownosc3}), 
the third term is not larger than $2\omega_{\mathrm{ev}}(n)$. To bound the fourth term in (\ref{Vader}), 
let us use the fact that $\|\mathbbm{1}+Z_A'\|\leq 2$ to write
\begin{eqnarray}
\fl\|\widetilde{X}_B(\mathbbm{1}+Z_A')(\mathbbm{1}-\widetilde{Z}_B)X_A'Z_A'\ket{\psi'}-\ket{\varphi'}\|&\leq& 
2\|\widetilde{X}_B(\mathbbm{1}-\widetilde{Z}_B)X_A'Z_A'\ket{\psi'}-(\mathbbm{1}+\widetilde{Z}_B)\ket{\psi'}\|.
\end{eqnarray} 
Subsequent usage of inequalities (\ref{nierownosc1}) and (\ref{nierownosc2}) to $\widetilde{Z}_B$
and $\widetilde{X}_B$ gives
\begin{eqnarray}
\fl\|\widetilde{X}_B(\mathbbm{1}+Z_A')(\mathbbm{1}-\widetilde{Z}_B)X_A'Z_A'\ket{\psi'}-\ket{\varphi'}\|\leq 16\sqrt{\varepsilon_1}+
2\|X_A'Z_A'(\mathbbm{1}-Z_A')X_A'\ket{\psi'}-(\mathbbm{1}+Z_A')\ket{\psi'}\|,\nonumber\\
\end{eqnarray}
which after double application of (\ref{nierownosc3}) yields
\begin{equation}
\|\widetilde{X}_B(\mathbbm{1}+Z_A')(\mathbbm{1}-\widetilde{Z}_B)X_A'Z_A'\ket{\psi'}-\ket{\varphi'}\|\leq 16\sqrt{\varepsilon_1}+2\omega_{\mathrm{ev}}(n).
\end{equation}
This together with previous estimations finally implies that  
\begin{equation}
\|\Phi(X_A'Z_A'\ket{\psi'}\ket{00})-\ket{\varphi'}X_AZ_B\ket{\phi_+}\|\leq 7\sqrt{\varepsilon_1}+2\omega_{\mathrm{ev}}(n).
\end{equation}
In a fully analogous way one can estimate the third term on the right-hand side of 
(\ref{czekolada}) 
\begin{equation}
\|\Phi(Z_A'X_B'\ket{\psi'}\ket{00})-\ket{\varphi'}Z_AX_B\ket{\phi_+}\|\leq  7\sqrt{\varepsilon_1}+2\omega_{\mathrm{ev}}(n).
\end{equation}
By plugging all these terms into (\ref{czekolada}) and then the 
resulting inequality together with (\ref{Homeland}) 
into (\ref{Orange}), one obtains 
\begin{equation}\label{SangreDeToro3}
\|\Phi(A_i'B'_j\ket{\psi'}\ket{00})-\ket{\varphi'}A_iB_j\ket{\phi_+}\|\leq 28\sqrt{\epsilon_1}+6\omega_{\mathrm{ev}}(n)+g_{\mathrm{ev}}+h_{\mathrm{ev}}.
\end{equation}

The terms from (\ref{Rob:thm1:2}) can be treated in almost exactly the same way, giving 
\begin{equation}\label{SangreDeToro}
\|\Phi(A'_i\ket{\psi'}\ket{00})-\ket{\varphi'}A_i\ket{\phi_+}\|\leq 12\sqrt{\varepsilon_1}+3\omega_{\mathrm{ev}}(n)+g_{\mathrm{ev}},
\end{equation}
while the estimation of the corresponding expression from \ref{Rob:thm1:2} follows from the application of 
inequality (\ref{nierownosc1}) to (\ref{SangreDeToro}), meaning that an additional error of 
$\sqrt{\varepsilon_1}$ has to be taken into account, which gives
\begin{equation}\label{SangreDeToro2}
\|\Phi(B'_j\ket{\psi'}\ket{00})-\ket{\varphi'}B_j\ket{\phi_+}\|\leq 13\sqrt{\varepsilon_1}+3\omega_{\mathrm{ev}}(n)+h_{\mathrm{ev}}.
\end{equation}
Finally, the case of $A_i'=B_j'=\mathbbm{1}$ has already been derived in (\ref{GranSangreDeToro}). 

The distance between the normalized state $\ket{\varphi}$ and the unnormalized one $\ket{\varphi'}$
is estimated in Lemma \ref{junklemma} to be 
\begin{equation}\label{states}
\|\ket{\varphi}-\ket{\varphi'}\|\leq \left(\frac{1}{2}+\sqrt{2}\right)\sqrt{\varepsilon_1}+\omega',
\end{equation}
where $\omega'$ is equal to $\omega_{ev}$ for an even number of inputs.

In order to obtain inequalities (\ref{Rob:thm1:1}) and 
complete the proof we use the triangle inequality for the vector norm to write
\begin{eqnarray}
\|\Phi(A_i'B_j'\ket{\psi'}\ket{00})-\ket{\varphi}A_iB_j\ket{\phi_+}\|&\leq& \|\Phi(A_i'B_j'\ket{\psi'}\ket{00})-\ket{\varphi'}A_iB_j\ket{\phi_+}\|\nonumber\\
&&+\|\ket{\varphi}-\ket{\varphi'}\|,
\end{eqnarray}
and then apply the previously determined inequalities (\ref{GranSangreDeToro}), (\ref{SangreDeToro3}), (\ref{SangreDeToro}), (\ref{SangreDeToro2}) and (\ref{states}).
All terms contributing to the functions $f(\varepsilon,n)$, $f_{B_j}(\varepsilon,n)$, $f_{A_i}(\varepsilon,n)$ and $f_{ij}(\varepsilon,n)$ scale at most as $O(n^2\sqrt{\epsilon})$. The more detailed analysis of the asymptotic behaviour of different contributions is discussed in Lemmas \ref{Rob:lem1} and \ref{lastlemma} in \ref{app:rob}.
\end{proof}
%

%\added{
Let us remark here that we have not checked whether the 
bounds (\ref{Rob:thm1:1})--(\ref{Rob:thm1:4}) are optimal
both in the distance from the maximal quantum violation $\varepsilon$
and the number of measurements $n$. Thus, it is still possible that these robustness bounds scale better that quadratically with the number of measurements. However, in order to determine such tighter bounds one would need in particular to optimize the above method over all SOS decomposition, which is certainly a difficult tasks.%}

\section{Randomness certification with the chained Bell inequalities}
It has been shown in Ref. \cite{DPA} that by exploiting the symmetry properties of the chained Bell inequality, one can certify two bits of randomness when the maximum quantum violation of this inequalities are obtained,
provided this maximal quantum violation is unique. However, a proof of the latter fact has not been known so far. Thus, our paper completes the result of Ref. \cite{DPA}.

Let us now provide an alternative way of certifying two bits of perfect randomness with the aid of the chained Bell inequality. For this purpose, we consider the following modification of the chained Bell inequality
\begin{equation}\label{newBell}
\widetilde{I}_{\mathrm{ch}}^n:=\mathcal{I}_{\mathrm{ch}}^n+
\langle A_1B_{n+1}\rangle\leq 2n-1
\end{equation}
in which Alice, as before, can measure one of $n$ observables $A_i$ while Bob has $n+1$ observables $B_i$ at his disposal, where $n$ is assumed to be even. It is not difficult to see that the maximal quantum violation of this inequality amounts to $\widetilde{B}_n^{\max}=B_n^{\max}+1$.

Let us now assume that $\ket{\psi}$ and $A_i$ and $B_i$ are the state and the measurements maximally violating (\ref{newBell}). Denoting then by
$\widetilde{\mathcal{B}}_n=\mathcal{B}_n+A_1\otimes B_{n+1}$ the corresponding Bell operator, one has $\langle\psi|(\widetilde{B}^{\max}_n\mathbbm{1}-\widetilde{\mathcal{B}}_n)|\psi\rangle=0$, which, owing to the fact that $\ket{\psi}$ also violates maximally the chained Bell inequality and that $B_n^{\max}$ is its maximal quantum violation, simplifies to
%
%\begin{equation}
$0=\langle\psi|(\mathbbm{1}-A_1\otimes B_{n+1})|\psi\rangle=(1/2)\langle\psi|(\mathbbm{1}-A_1\otimes B_{n+1})^2|\psi\rangle$,
%\end{equation}
%
where the second equality is a consequence of the fact that
$A_1$ and $B_{n+1}$ are unitary and hermitian. This implies that
\begin{equation}\label{equation}
A_1\ket{\psi}=B_{n+1}\ket{\psi}.
\end{equation}
This property implies in particular that $\langle B_{n+1}\rangle=\langle A_1\rangle$, which, taking into account the fact that
for the maximal quantum violation of the chained Bell inequality
$\langle A_i\rangle=0$ for any $i=1,\ldots,n$, implies $\langle B_{n+1}\rangle=0$.
In a quite analogous way we can now prove that
the expectation value $\langle A_{n/2+1}B_{n+1}\rangle=\langle\psi| A_{n/2+1}\otimes B_{n+1}|\psi\rangle$ vanishes. Exploiting Eq. (\ref{equation}), we can rewrite it as
$\langle\psi| A_{n/2+1}\otimes B_{n+1}|\psi\rangle=\langle\psi|A_{n/2+1}A_1|\psi\rangle$. Then, due to the fact that the expectation value
$\langle\psi| A_{n/2+1}\otimes B_{n+1}|\psi\rangle$ is real and
both operators $A_{n/2+1}$ and $B_{n+1}$ are hermitian, which means
that $\langle\psi|A_{n/2+1}A_1|\psi\rangle=\langle\psi|A_1A_{n/2+1}|\psi\rangle$,
this can be further rewritten as
\begin{equation}
\langle A_{n/2+1}B_{n+1}\rangle=
%
%\frac{1}{2}\langle\psi|A_{n/2+1}A_1|\psi\rangle+\frac{1}{2}\langle\psi|A_1A_{n/2+1}|%\psi\rangle=
%
\frac{1}{2}\langle\psi|\{A_1,A_{\frac{n}{2}+1}\}|\psi\rangle.
\end{equation}
We have already proven that if $\ket{\psi}$ and $A_i$ and $B_i$ violate maximally the chained Bell inequality, then $\{A_1,A_{n/2+1}\}\ket{\psi}=0$ which implies that
$\langle A_{n/2+1}B_{n+1}\rangle=0$, which together with
$\langle A_1\rangle=\langle B_{n+1}\rangle=0$ mean finally that
\begin{equation}\label{PinkMartini}
p(a,b|A_{\frac{n}{2}+1},B_{n+1})=\frac{1}{4}
\end{equation}
with $a,b=0,1$. All this proves that any probability distribution
$p(a,b|A_i,B_j)$ with $i=1,\ldots,n$ and $j=1,\ldots,n+1$ maximally violating
the modified chained Bell inequality (\ref{newBell}) is such that
all outcomes of the pair of measurements $A_{n/2+1},B_{n+1}$ are equiprobable
(\ref{PinkMartini}) and thus perfectly random, meaning that (\ref{newBell})
certifies two bits of perfect randomness. %$\added{(uniqueness of this extended prob. distribution is not necessary)}.

The intuition behind the above approach is very simple.
At the maximal quantum violation of (\ref{newBell}) the measurement
$B_{n+1}$ must be ``parallel'' to $A_1$ [cf. Eq. (\ref{equation})].
Therefore it is ``orthogonal'' to $A_{n/2+1}$ as the latter is orthogonal to $A_1$,
meaning that $\langle A_{n/2+1}B_{n+1}\rangle=0$ which is basically what we need.
It is worth noticing that in the even $n$ case all pairs $A_{1+i},A_{n/2+i}$ with $i=1,\ldots,n/2-1$ of Alice's observables are orthogonal, and therefore
our argument can be extended to any pair $A_{n/2+i},B_{n+1}$, that is,
$\langle A_{n/2+i},B_{n+1}\rangle=0$ provided the Bell inequality
$\mathcal{I}_{\mathrm{ch}}^n+\langle A_{n/2+i}B_{n+1}\rangle\leq 2n-1$
is maximally violated. Unfortunately, this approach does not work in the odd $n$
case as no pair of observables at Alice's or Bob's sides are orthogonal.

\section{Discussion}

In this work, we developed a scheme for self-testing the maximally entangled state of two qubits using the chained Bell inequalities. Since our results hold for any number of inputs, this allows to self-test measurements on the whole $XZ$ plane of the Bloch sphere. Some of the previous self-testing techniques found an application for blind quantum computation protocols (See \cite{BQC}). The fact that chained Bell inequalities involve and certify a quite large class of measurements makes this self-testing protocol a good candidate for some future application in blind quantum computation processes. Beyond their interest as a protocol in quantum information processing, our results also have fundamental implications, since they prove the uniqueness of the maximal violation of the chained Bell inequalities. In \cite{DPA}, this property was assumed to be true to argue maximal randomness certification in Bell tests: with our proof, their results are now confirmed.
Contrary to the expectations, {when increasing the number of measurements, the robustness of our protocol diminishes. An interesting open question is to see whether it is possible to improve this scaling. Another open question concerns chained Bell inequalities with more outcomes: can they also be useful for self-testing? If so, one could also make use of these results for certifying random \textit{dits} in systems of dimension larger than two.

\section*{Acknowledgements}
This work is supported by the ERC CoG project QITBOX and ERC AdG project OSYRIS, the AXA Chair in Quantum Information Science, the Spanish MINECO (Severo Ochoa Grant SEV-2015-0522 and FOQUS FIS2013-46768-P), the Generalitat de Catalunya (SGR875) and The John Templeton Foundation. I\v{S} also acknowledges the support of Ministry
of Science of Montenegro (Physics of Nanostructures, Contract No 01-682).

\appendix
\section{Proving the SOS decompositions}\label{app:SOS}

Here we provide more detailed explanation of the SOS decompositions (\ref{eq:genSOS1st}) and (\ref{eq:SOSnBell2ndorder}). 

To verify the validity of the the first decomposition one expands
the first sum of its right-hand side and notices that apart from the 
terms forming the shifted Bell operator $B_n^{\max}\mathbbm{1}-\mathcal{B}_n$ there are some
additional terms of the form $B_kB_{k+1}$. These are cancelled out 
by the same terms appearing in the second sum on the right-hand side of Eq. (\ref{eq:genSOS1st}). 
The only trouble one has to face in reducing all the remaining terms to the shifted Bell operator is to prove that the coefficient multiplying the identity operator $\mathbbm{1}$ is exactly $2n\cos(\pi/2n)$. Let us now prove that indeed this is the case. To this end, we write this coefficient as
\begin{equation}\label{Delta}
\Delta = \cos{\frac{\pi}{2n}}\left[n + \frac{n}{2\cos^2(\pi/2n)} + \Delta_{\alpha}+\Delta_{\beta}+\Delta_{\gamma}\right],
%
%\sum_{i=1}^{n-2}(\alpha_i^2 + \beta_i^2 + \gamma_i^2)
\end{equation}
where 
\begin{equation}
\Delta_{\omega}=\sum_{i=1}^{n-2}\omega^2_i
\end{equation}
with $\omega=\alpha,\beta,\gamma$. Recall that the coefficients $\alpha_i$, $\beta_i$ and $\gamma_i$ are defined in
Eqs. (\ref{eq:coefficientsa}), (\ref{eq:coefficientsb}) and (\ref{eq:coefficientsg}). Let us now compute each 
term $\Delta_{\omega}$ separately, starting from $\Delta_{\alpha}$. Exploiting Eq. (\ref{eq:coefficientsa}) we can write
\begin{eqnarray}\label{Tempura}
\Delta_{\alpha} & = \frac{1}{4\cos^2(\pi/2n)}\sum_{i=1}^{n-2}\left[\frac{\sin^2(\pi/n)}{\sin(i\pi/n)\sin[(i+1)\pi/n]}\right] \nonumber\\
& = \frac{\sin(\pi/n)}{4\cos^2(\pi/2n)}\sum_{i=1}^{n-2}\left[\frac{\cos(i\pi/n)}{\sin(i\pi/n)} - \frac{\cos[(i+1)\pi/n]}{\sin[(i+1)\pi/n]}  \right]\nonumber\\
& = \frac{\sin(\pi/n)}{4\cos^2(\pi/2n)}\sum_{i=1}^{n-2}\left\{\cot\left(\frac{i\pi}{n}\right) - \cot\left[\frac{(i+1)\pi}{n}\right]  \right\}.
\end{eqnarray}
Now, we utilize the fact that 
\begin{equation}
\sum_{i=1}^{n-1}\cot\left(\case{\pi i}{n}\right)=0
\end{equation}
which implies that
\begin{eqnarray}\label{Sushi}
\sum_{i=1}^{n-2}\cot(\case{i\pi}{n}) = \cot(\case{\pi}{n}),
\qquad\sum_{i=1}^{n-2}\cot(\case{(i+1)\pi}{n}) = -\cot(\case{\pi}{n}),
\end{eqnarray}
Substituting Eq. (\ref{Sushi}) into Eq. (\ref{Tempura}) one finds that 
\begin{equation}\label{Deltaa}
\Delta_{\alpha}=\frac{\cos(\pi/n)}{2\cos^2(\pi/2n)}.
\end{equation}

Let us then compute $\Delta_{\beta}$. Using Eq. (\ref{eq:coefficientsb}), it 
can be explicitly written as
\begin{equation}\label{eq:D2}
\Delta_{\beta}   = \frac{1}{4\cos^2(\pi/2n)}\left[\sum_{i=1}^{n-2}\frac{\sin[(i+1)\pi/n]}{\sin(i\pi/n) }\right],
\end{equation}
which with the aid of the elementary trigonometric property that 
$\sin(x+y)=\sin x\cos y+\cos x\sin y$, rewrites as
\begin{equation}
\Delta_{\beta}=\frac{1}{4\cos^2(\pi/2n)}\left[(n-2)\cos(\case{\pi}{n})+ \sin(\case{\pi}{n})\sum_{i=1}^{n-2}\cot(\case{i\pi}{n})\right].
\end{equation}
This, by virtue of (\ref{Sushi}), finally gives 
\begin{equation}\label{Deltab}
\Delta_{\beta}=\frac{(n-1)\cos(\pi/n)}{4\cos^2(\pi/2n)}.
\end{equation}
%
%\begin{eqnarray}
% = \frac{1}{4\cos^2(\frac{\pi}{2n})}\left[\sum_{i=1}^{n-2}\frac{\sin(\frac{i\pi}{n})\cos(\frac{\pi}{n}) %+\cos(\frac{i\pi}{n})\sin(\frac{\pi}{n})}{\sin(\frac{i\pi}{n}} \right]\nonumber\\
% = \frac{1}{4\cos^2(\frac{\pi}{2n})}\left[(n-2)\cos(\frac{\pi}{n})+ \sin(\frac{\pi}{n})\sum_{i=1}^{n-2}\cot(\frac{i\pi}{n})\right]\nonumber\\
% = \frac{1}{4\cos^2(\frac{\pi}{2n})}\left[(n-2)\cos(\frac{\pi}{n})+ \sin(\frac{\pi}{n})\cot(\frac{\pi}{n}) %\right]\nonumber\\
% = \frac{(n-1)\cos(\frac{\pi}{n})}{4\cos^2(\frac{\pi}{2n})}
%\end{eqnarray}
%

Let us finally compute $\Delta_{\gamma}$. From (\ref{eq:coefficientsg}) it can be 
written explicitly as
\begin{equation}
\Delta_{\gamma}  =  \frac{1}{4\cos^2(\pi/2n)}\left[\sum_{i=1}^{n-2}\frac{\sin(i\pi/n)}{\sin[(i+1)\pi/n]} \right].
\end{equation}
Writing then $\sin(i\pi/n)=sin[(i+1-1)\pi/n]$ and using again the above trigonometric identity, one obtains
\begin{equation}
\Delta_{\gamma}=\frac{1}{4\cos^2(\pi/2n)}\left\{(n-2)\cos(\case{\pi}{n}) - \sin(\case{\pi}{n})\sum_{i=1}^{n-2}\cot\left[\case{(i+1)\pi}{n}\right] \right\}, 
\end{equation}
which, taking into account Eq. (\ref{Sushi}), simplifies to 
\begin{equation}\label{Deltag}
\Delta_{\gamma}=\frac{(n-1)\cos(\pi/n)}{4\cos^2(\pi/2n)}.
\end{equation}
Plugging then Eq. (\ref{Deltaa}), (\ref{Deltab}) and (\ref{Deltag}) into (\ref{Delta}) and using some elementary properties of the trigonometric functions, one eventually 
obtains $\Delta=2n\cos(\pi/2n)$.

To confirm validity of the second degree SOS decomposition (\ref{eq:SOSnBell2ndorder}) we follow similar argumentation. The first parenthesis on the right hand side of (\ref{eq:SOSnBell2ndorder}) introduces terms that up to some multiplicative factors belong to the following set $\{ \mathbbm{1}, A_iB_i, A_iB_{i-1}, A_iA_{i+1}, B_{i}B_{i+1}, A_iA_jB_kB_l \}$. The terms $A_iA_jB_kB_l$ are directly cancelled out by the same terms stemming 
from the second and the third parenthesis. Then, the terms $A_iA_{i+1}$ and $B_{i}B_{i+1}$ enter with the coefficient $2/[8n\cos(\pi/2n)]$ and together with the same terms resulting from the second parenthesis and entering with the coefficient $(n-2)/[8n\cos(\pi/2n)]$ they are cancel out by those resulting from the third line of (\ref{eq:SOSnBell2ndorder}).
The terms $A_iB_i$ and $A_iB_{i-1}$ give rise to the shifted Bell operator, and, finally, the identity operator $\mathbbm{1}$ is multiplied by the following expression
\begin{equation}
 \frac{1}{8n\cos(\pi/2n)}\left\{\left[8n^2\cos^2(\case{\pi}{2n}) + 4n\right] + 4n(n-2) + 4n\right\} + \frac{n\cos(\pi/n)}{2\cos^2(\pi/2n)}
\end{equation}
which after some movements simplifies to 
$2n\cos(\pi/2n)$. This is exactly the multiplicative factor of 
identity operator in the shifted Bell operator.

\section{Exact case}\label{app:exact}

Here we present detailed proofs of the anticommutation relation (\ref{AntiExact}), and also of some auxiliary relations for the measurements $A_i$ and $B_i$. Before we proceed let us note that in some of the following expressions operators might be indexed by any integer (not just from the set $\{ 1, \dots , n \} $), and in those cases we use the notation $C_{n+i} = -C_i$ and $C_{-i} = -C_{n-i}$. The intuition for this notation can be found on Bloch sphere representation of the measurements (\ref{fig:measurements}), where we can see that if one would draw the next measurement after $C_n$, and note it as $C_{n+1}$ it would be parallel to $-C_1$, and similarly for any $C_{n+i}$. 

Let us start by proving the following lemma.

\begin{lem}\label{Exact:lem1}
Let $\{\ket{\psi'},A_i',B_i'\}$ be the pure state and the measurements realizing the maximal quantum violation of the chained Bell inequalities. Then, the following identities are true:
\begin{equation}\label{Exact:lem1:1}
A'_i\ket{\psi'}=\frac{B_i'+B_{i-1}'}{2\cos(\pi/2n)}\ket{\psi'} \equiv B'_{i-1,i}\ket{\psi'}
\end{equation} 
for $i=1,\ldots,n$,
\begin{equation}\label{Exact:lem1:2}
(\alpha_iC_j+\beta_iC_{i+j}+\gamma_iC_{i+j+1})\ket{\psi'}=0
\end{equation}
for $i=1,\ldots,n-2$, $j=1,\ldots,n$ and $C=A',B'$, and
\begin{equation}\label{Exact:lem1:3}
(A_i'B_i'-A_{i+1}'B_{i+1}')\ket{\psi'}=0
\end{equation}
\begin{equation}\label{Exact:lem1:4}
(A_i'B_{i-1}'-A_{i+1}'B_{i}')\ket{\psi'}=0
\end{equation}
for $i=1,\ldots,n$.
\end{lem}
\begin{proof}From the fact that $\ket{\psi'}$
and $A_i'$ and $B_i'$ violate the chained Bell inequality maximally it follows that 
$\langle\psi|(B_n^{\max}\mathbbm{1}-\mathcal{B}_n)\ket{\psi'}=0$. Now, the 
first SOS decomposition (\ref{eq:genSOS1st}) for the operator 
$B_n^{\max}\mathbbm{1}-\mathcal{B}_n$ implies Eqs. (\ref{Exact:lem1:1})
and (\ref{Exact:lem1:2}), while the second one implies Eqs. (\ref{Exact:lem1:3})
and (\ref{Exact:lem1:4})
\end{proof}

%\subsection{Even number of measurements}\label{appeven}

\begin{lem}\label{Exact:Anti}
Let $\{\ket{\psi'},A_i',B_i'\}$ be the pure state and the measurements realizing the maximal quantum violation of the chained Bell inequalities. Then, the following relations are true:
\begin{equation}\label{Lemma1:e}
\{A_1',A_{\frac{n}{2}+1}'\}\ket{\psi'}=0
\end{equation}
for even $n$, and
\begin{equation}\label{Lemma1:o}
\{A_1',A_{\frac{n+1}{2}}'+A_{\frac{n+3}{2}}'\}\ket{\psi'}=0
\end{equation}
for odd $n$.
\end{lem}
\begin{proof}We prove the even and odd $n$ case separately.

\textbf{Even number of measurements.} Let us begin by noting that 
by setting $j=k-i$ with $k=1,\ldots,n$ in (\ref{Exact:lem1:2}), one obtains
\begin{equation}\label{Lemma1:2}
(\alpha_iC_{k-i}+\beta_iC_{k}+\gamma_iC_{k+1})\ket{\psi'}=0.
\end{equation}
On the other hand, by shifting $i\to n-i-1$ and setting $j=k+i+1$, we arrive at
\begin{equation}\label{Lemma1:3}
(\alpha_{n-i-1}C_{k+i+1}+\beta_{n-i-1}C_{k+n}+\gamma_{n-i-1}C_{k+n+1})\ket{\psi'}=0,
\end{equation}
which, by noting that $C_{k+n}=-C_k$ for any $k=1,\ldots,n-1$, $\alpha_{n-i-1}=\alpha_i$ and $\beta_{n-i-1}=-\gamma_i$ for any $i=1,\ldots,n-2$, can further be simplified to
\begin{equation}\label{Lemma1:4}
(\alpha_{i}C_{k+i+1}+\gamma_{i}C_k+\beta_{i}C_{k+1})\ket{\psi'}=0.
\end{equation}
After summing Eqs. (\ref{Lemma1:2}) and (\ref{Lemma1:4}) and performing some straightforward manipulations
we finally obtain
\begin{equation}\label{Lemma1:5}
(C_{k-i}+C_{k+i+1})\ket{\psi'}=\xi_iC_{k,k+1}\ket{\psi'},
\end{equation}
where we denoted $\xi_i=2\cos[(2i+1)\pi/2n]$ and $C_{k,k+1}=(C_k+C_{k+1})/[2\cos(\pi/2n)]$. Finally, setting $k=0$ in Eq. 
(\ref{Lemma1:4}) and $k=n$ in Eq. (\ref{Lemma1:2}) and subtracting the resulting equations one from another we have
\begin{equation}\label{Lemma1:6}
(C_{i+1}-C_{n-i})\ket{\psi'}=\xi_i C_{1,-n}\ket{\psi'},
\end{equation}
where we have denoted $C_{1,-n}=(C_1-C_n)/[2\cos(\pi/2n)]$.

Having all these auxiliary identities at hand, we are now in position to 
prove Eq. (\ref{Lemma1:e}). To this end, we first rewrite its left-hand side as
\begin{eqnarray}
\label{eq:ACeven2}
(A_1'A_{\frac{n}{2}+1}'+A_{\frac{n}{2}+1}'A_1')\ket{\psi'}&=&\left(A_1'B_{\frac{n}{2},\frac{n}{2}+1}' + A_{\frac{n}{2}+1}'B_{1,-n}' \right)\ket{\psi'}\nonumber\\
&=& \frac{1}{\xi_{\frac{n}{2}-1}}\left[A_1'(B_1' + B_{n}') + A_{\frac{n}{2}+1}'(B_{\frac{n}{2}}' - B_{\frac{n}{2}+1}') \right]\ket{\psi'},\nonumber\\
\end{eqnarray}
where the first equality was obtained with the aid of the identity (\ref{Exact:lem1:1})
for $i=n/2+1$,
%
%\begin{equation}\label{Lemma1:8}
%A_i\ket{\psi'}=\frac{B_i-B_{i-1}}{2\cos(\pi/2n)}\ket{\psi'}\qquad (i=1,\ldots,n),
%\end{equation} 
%
%stemming from the SOS decomposition (\ref{eq:genSOS1st}), 
%
while the second one follows from Eqs. (\ref{Lemma1:5}) and (\ref{Lemma1:6}). Then, 
the formulas (\ref{Exact:lem1:3}) and (\ref{Exact:lem1:4}) imply that 
%
%from the second SOS decomposition (\ref{eq:SOSnBell2ndorder}) we see that $(A_iB_i - A_{i+1}B_{i+1})\ket{\psi'} = 0$ and $(A_iB_{i-1} - A_{i+1}B_{i})%%%%%%\ket{\psi'} = 0$ for all $i$, which imply, respectively, that 
%
\begin{equation}\label{Lemma1:9}
(A_1'B_1' - A_{j+1}'B_{j+1}')\ket{\psi'} =\sum_{i=1}^{j}(A_i'B_i' - A_{i+1}'B_{i+1}')\ket{\psi'}= 0
\end{equation}
%'
and
\begin{equation}\label{Lemma1:10}
(A_1'B_{n}' + A_{j+1}'B_{j}')\ket{\psi'}=\sum_{i=1}^{j}(A_i'B_{i-1}' - A_{i+1}'B_{i}')\ket{\psi'} = 0
\end{equation}
hold for any $j=1,\ldots,n$. After setting $j=n/2$ in the latter identities and inserting them 
into Eq. (\ref{eq:ACeven2}) we eventually obtain
(\ref{Lemma1:e}).

\textbf{Odd number of measurements.} Before passing to the anticommutation relation (\ref{Lemma1:o}), we need some auxiliary relations for the measurements $A_i'$ and $B_i'$. In order to derive the first one, we shift $k\to k-1$ in Eq. (\ref{Lemma1:4}) and add the resulting equation to Eq. (\ref{Lemma1:2}), obtaining
\begin{equation}\label{Lemma1:11}
(C_{k+i}+C_{k-i})\ket{\psi'}=-2\frac{\beta_i}{\alpha_i}C_k-\frac{\gamma_i}{\alpha_i}(C_{k-1}+C_{k+1})\ket{\psi'}.
\end{equation}
Then, setting $i=1$ and shifting $j\to j-1$ in Eq. (\ref{Exact:lem1:2}) we arrive at
\begin{equation}\label{Lemma1:12}
(C_{j+1}+C_{j-1})\ket{\psi'}=2\cos\left(\frac{\pi}{n}\right)C_j\ket{\psi'},
\end{equation}
which after being plugged into Eq. (\ref{Lemma1:11}) gives rise to the following identity
\begin{equation}\label{Lemma1:13}
(C_{k+i}+C_{k-i})\ket{\psi'}=\zeta_iC_k\ket{\psi'},
\end{equation}
where $\zeta_i=2\cos(i\pi/n)$.

Then, by setting $j=(n-1)/2$ in Eqs. (\ref{Lemma1:9}) and (\ref{Lemma1:10}) and adding the resulting equations 
we obtain 
\begin{equation}\label{Lemma1:14}
A_1'(B_1'+B_n')\ket{\psi'}=A_{\frac{n+1}{2}}'(B_{\frac{n+1}{2}}'-B_{\frac{n-1}{2}}')\ket{\psi'},
\end{equation}
which can be further simplified by using Eq. (\ref{Lemma1:13}) with $i=(n-1)/2$ and $k=n$, giving
\begin{equation}\label{Lemma1:15}
A_1'(B_1'+B_n')\ket{\psi'}=\zeta_{\frac{n-1}{2}}A_{\frac{n+1}{2}}'B_n'\ket{\psi'}.
\end{equation}
Analogously, by setting $j=(n+1)/2$ in Eqs. (\ref{Lemma1:9}) and (\ref{Lemma1:10}) and adding them, one obtains
\begin{equation}\label{Lemma1:16}
A_1'(B_1'+B_n')\ket{\psi'}=A_{\frac{n+3}{2}}'(B_{\frac{n+3}{2}}'-B_{\frac{n+1}{2}}')\ket{\psi'},
\end{equation}
which, after application of Eq. (\ref{Lemma1:13}) with $i=(n-1)/2$ and $k=n+1$, further simplifies to 
\begin{equation}\label{Lemma1:17}
A_1'(B_1'+B_n')\ket{\psi'}=-\zeta_{\frac{n-1}{2}}A_{\frac{n+3}{2}}'B_1'\ket{\psi'}.
\end{equation}

Now, we can rewrite the left-hand side of the anticommutation relation
Eq. (\ref{Lemma1:o}) as
\begin{eqnarray}
\label{eq:ACodd1}
\fl\left\{ A_1',A_{\frac{n+1}{2}}' + A_{\frac{n+3}{2}}' \right\} \ket{\psi'} &=& \frac{1}{2\cos{\case{\pi}{2n}}}\left[A_1'(B_{\frac{n-1}{2}}' + 2B_{\frac{n+1}{2}}' + B_{\frac{n+3}{2}}') + (A_{\frac{n+1}{2}}' + A_{\frac{n+3}'{2}})(B_1' - B_{n}') \right]\ket{\psi'}\nonumber\\
&&\hspace{-2cm}=\frac{1}{2\cos\case{\pi}{2n}}\left[A_1'\left(B_{\frac{n-1}{2}}' + B_{\frac{n+3}{2}}' + 2\frac{B_1' + B_{n}'}{\zeta_{(n-1)/2}}\right)+ (A_{\frac{n+1}{2}}' + A_{\frac{n+3}{2}}')(B_1' - B_{n}') \right]\ket{\psi'},\nonumber\\
\end{eqnarray}
where first equality stems from Eq. (\ref{Exact:lem1:1}) and to obtain the second one 
we have utilized Eq. (\ref{Lemma1:13}) with $i=(n-1)/2$ and $k=(n+1)/2$. 
Then, expressions (\ref{Lemma1:15}) and (\ref{Lemma1:17}) lead us to
\begin{equation}
\label{eq:ACodd2}
\fl\left\{ A_1',A_{\frac{n+1}{2}}' + A_{\frac{n+3}{2}}' \right\} \ket{\psi'} = \frac{1}{2\cos{(\pi/2n)}}(A_1'B_{\frac{n-1}{2}}' + A_1'B_{\frac{n+3}{2}}' + A_{\frac{n+1}{2}}'B_{1}' - A_{\frac{n+3}{2}}'B_{n}')\ket{\psi'}.
\end{equation}
Exploiting once more Eq. (\ref{Lemma1:13}) one obtains the following equalities
\begin{equation}
A_1'\ket{\psi'} = \frac{1}{\zeta_{\frac{n-1}{2}}}(A_{\frac{n+1}{2}}'-A_{\frac{n+3}{2}}')\ket{\psi'},\quad 
B_1'\ket{\psi'} = \frac{1}{\zeta_{\frac{n-1}{2}}}(B_{\frac{n+1}{2}}'-B_{\frac{n+3}{2}}')\ket{\psi'},
\end{equation}
and
\begin{equation}
B_{n}'\ket{\psi'} = \frac{1}{\zeta_{\frac{n-1}{2}}}(B_{\frac{n+1}{2}}'-B_{\frac{n-1}{2}}')\ket{\psi'}
\end{equation}
whose application to Eq. (\ref{eq:ACodd2}) allows one to rewrite it as
\begin{eqnarray}
\label{eq:ACodd22}
\fl\left\{ A_1',A_{\frac{n+1}{2}}' + A_{\frac{n+3}{2}}' \right\} \ket{\psi'} = \frac{1}{2\zeta_{\frac{n-1}{2}}\cos{\case{\pi}{2n}}}\left(A_{\frac{n+1}{2}}'B_{\frac{n-1}{2}}' -A_{\frac{n+3}{2}}'B_{\frac{n+1}{2}}'+A_{\frac{n+1}{2}}'B_{\frac{n+1}{2}}'-A_{\frac{n+3}{2}}'B_{\frac{n+3}{2}}'\right)\ket{\psi'}.\nonumber\\
\end{eqnarray}
To complete the proof it suffices to make use of the equalities (\ref{Exact:lem1:3}) and (\ref{Exact:lem1:4})
with $j=(n+1)/2$.
\end{proof}

\begin{lem}
\label{Structure}
Let $\{\ket{\psi'},A'_i,B'_i\}$ realize the maximal quantum violation of the chained Bell inequality. 
Then, for even $n$:
\begin{eqnarray}
%\label{eq:ABstructure}
\label{eq:ABstructure1}A'_i\ket{\psi'} &=& \left(s_iA'_{\frac{n}{2}+1} + c_iA'_1\right)\ket{\psi'}, \\
\label{eq:ABstructure2}B'_i\ket{\psi'} &=& \left(s_i'B'_{\frac{n}{2},\frac{n}{2}+1} + c_i'B'_{1,-n}\right)\ket{\psi'},
\end{eqnarray}
while for odd $n$:
\begin{eqnarray}
%\label{eq:ABstructure}
\label{eq:ABstructure1o}A_i'\ket{\psi'} &=& \left\{s_iA_{\frac{n+1}{2},\frac{n+3}{2}}' + c_iA_1'\right\}\ket{\psi'}, \\
\label{eq:ABstructure2o}B_i'\ket{\psi'} &=& \left\{s_i'B_{\frac{n+1}{2}}' + c_i'B_{1,-n}'\right\}\ket{\psi'},
\end{eqnarray}
are valid for any $i=1,\ldots,n$. Symbols $s_i$, $c_i$, $s_i'$ and $c_i'$ are defined in Eq. (\ref{eq:ChainedMeasurementsA}).
\end{lem}
\begin{proof}Let us begin with the even $n$ case. By setting 
$k=1+n/2$ and shifting $i\to 1-i+n/2$ in Eq. (\ref{Lemma1:13}) one obtains
\begin{equation}\label{Lemma2:1}
(C_i-C_{2-i})\ket{\psi'}=\zeta_{\frac{n}{2}+1-i}C_{\frac{n}{2}+1}\ket{\psi'}.
\end{equation}
for $i=1,\ldots,n/2$, where we have additionally exploited the fact that $C_{n+i}=-C_i$ and $C_{-i} = -C_{n-i}$ for any $i$. To prove Eq. (\ref{Lemma2:1}) for $i=n/2+1,\ldots,n/2$ one has to use (\ref{Lemma1:2}) but coefficients $\alpha_i$, $\beta_i$ and $\gamma_i$ are not defined for $i < 0$. However, once Eq. (\ref{Lemma2:1}) is derived for $i < n/2 +1$, it is easy to note that the cases when $i > n/2 + 1$ are already contained in the proof. This is due to the fact that any expression obtained when $i > n/2 + 1$, is the same as the expression proved for  $n+2-i < n/2 + 1$.\\
On the other hand, fixing $k=1$ and shifting $i\to i-1$ in Eq. (\ref{Lemma1:13}), one can deduce the following
equality
\begin{equation}\label{Lemma2:2}
(C_i+C_{2-i})\ket{\psi'}=\zeta_{i-1}C_1\ket{\psi'}. 
\end{equation}
with $i=2,\ldots n$. For $i=1$ the equation is trivial. Adding Eqs. (\ref{Lemma2:1}) and (\ref{Lemma2:2}) and recalling that $\zeta_i=2\cos(i\pi/n)$
one obtains Eq. (\ref{eq:ABstructure1}).

In order to prove the second identity (\ref{eq:ABstructure2}), 
we fix $k=n/2$ and shift $i\to n/2-i$ in Eq. (\ref{Lemma1:5}) which leads us to 
\begin{equation}\label{Lemma2:3}
(C_i+C_{n-i+1})\ket{\psi'}=\xi_{\frac{n}{2}-i}C_{\frac{n}{2},\frac{n}{2}+1}\ket{\psi'}.
\end{equation}
This equation is satisfied for all $i =  1, \ldots , n$, but it could formally be derived only when $i < n/2$. The cases $i = n/2, n/2+1$ are trivially satisfied. Similarly to the discussion following Eq. (\ref{Lemma2:1}) it is easy to check that for every $i > n/2+1$ Eq. (\ref{Lemma2:3}) is the same as for the case $n+1-i < n/2$, which has been formally proven.

Now we note that by shifting $i\to i-1$ in Eq. (\ref{Lemma1:6}), one obtains the following equation
\begin{equation}\label{Lemma2:4}
(C_i-C_{n-i+1})\ket{\psi'}=\xi_{i-1}C_{1,-n}\ket{\psi'}, 
\end{equation}
which when combined with Eq. (\ref{Lemma2:3}) directly implies Eq. (\ref{eq:ABstructure2}), completing the proof.

Now we move to the odd $n$ case. First in Eq. (\ref{Lemma1:5}) we fix $k = (n+1)/2$ and shift $i \rightarrow (n+1)/2 - i$ to get
\begin{equation}\label{Lemma2:5}
(C_i + C_{n+2-i})\ket{\psi'} = \xi_{\frac{n+1}{2} - i}C_{\frac{n+1}{2},\frac{n+3}{2}}\ket{\psi'}.
\end{equation}
This equation is consistent for all $i = 1 , \dots , n$, with the clarification exactly the same as in the discussion following Eq. (\ref{Lemma2:3}). Next step is to plug $k = 1$ and $i \rightarrow i -1$ in Eq. (\ref{Lemma1:13}) which together with $C_{2-i} = -C_{n+2-i}$ gives
\begin{equation}\label{Lemma2:6}
(C_i - C_{n+2-i})\ket{\psi'} = \zeta_{i-1}C_{1}\ket{\psi'}
\end{equation}
By adding Eqs. (\ref{Lemma2:5}) and (\ref{Lemma2:6}) and using some elementary trigonometric identities we obtain \ref{eq:ABstructure1o}.
We proceed by fixing $k = (n+1)/2$ and shifting $i \rightarrow (n+1)/2 - i$ in Eq. (\ref{Lemma1:13}) to obtain
\begin{equation}\label{Lemma2:7}
(C_i + C_{n+1-i})\ket{\psi'} = \zeta_{\frac{n+1}{2} - i}C_{\frac{n+1}{2}}\ket{\psi'},
\end{equation}
satisfied for all $i = 1,\dots , n$ in the same way as Eq. (\ref{Lemma2:1}). To get Eq. (\ref{eq:ABstructure2o}) and complete the proof to Eq. (\ref{Lemma2:7}) we add
\begin{equation}\label{Lemma2:8}
(C_i - C_{n+1-i})\ket{\psi'} = \xi_{i-1}C_{1,-n}\ket{\psi'}
\end{equation}
 which is obtained by shifting $i \rightarrow i-1$ in Eq. (\ref{Lemma1:6})
\end{proof}

\section{Robustness}\label{app:rob}

Here we present detailed proofs of the relations exploited in Section \ref{3}.
We begin with the approximate version of Lemma \ref{Exact:lem1}.

\begin{lem} 
\label{lemmaRobust}
Let $\ket{\psi'}$ and $\{A_i', B_i'\}$ be the state and the measurements violating the chained Bell inequality by $B_n^{\max} - \epsilon$. Then, the following relations are satisfied:
\begin{equation}\label{eq:rob1}
\| (A_i' - B'_{i-1,i})\ket{\psi'} \| \leq \sqrt{\frac{\varepsilon}{\cos(\pi/2n)}}\equiv\sqrt{\epsilon_1}
\end{equation}
for $i=1,\ldots,n$,
\begin{equation}\label{eq:rob2}
\| (\alpha_{i} B_{j}' + \beta_{i} B_{i + j}' + \gamma_{i} B_{ i + j + 1}')\ket{\psi'} \| \leq \sqrt{\varepsilon_1}
\end{equation}
for $i=1,\ldots,n-2$ and $j=1,\ldots,n$, and
\begin{eqnarray}
\label{eq:rob3}
\| (A_i' \otimes B_i' - A_{i+1}' \otimes B_{i+1}')\ket{\psi'} \| \leq \sqrt{8n\cos\frac{\pi}{2n}\epsilon}\equiv\sqrt{n\varepsilon_2},\\
\label{eq:rob4}
\| (A_i' \otimes B_{i-1}' - A_{i+1}' \otimes B_{i}')\ket{\psi'} \| \leq \sqrt{n\varepsilon_2}
\end{eqnarray}
for $i=1,\ldots,n$.
\end{lem}
\begin{proof}
All equations follow directly from SOS decompositions. When a chained Bell inequality is violated by $2n\cos[\pi/2n] - \epsilon$, from (\ref{eq:BellSOS}) it follows that $\sum_i \bra{\psi'}P_i^2\ket{\psi'} = \epsilon$ and consequently $\vert \vert P_i\ket{\psi'} \vert \vert \leq \sqrt{\epsilon}$ for all $i$. The expressions given by equations (\ref{eq:rob1}) and (\ref{eq:rob2}) are identified in the first degree SOS decomposition (\ref{eq:genSOS1st}) (note the explanation after the equation), while the expressions bounded in equations (\ref{eq:rob3}) and (\ref{eq:rob4}) are the part of the second degree SOS decomposition (\ref{eq:SOSnBell2ndorder}).
\end{proof}

We can then prove the approximate version of Lemma \ref{Exact:Anti} .

\begin{lem}\label{Rob:lem1}
Let $\{\ket{\psi'},A_i',B_i'\}$ be the state and the measurements violation the chained Bell inequality by $B_n^{\max}-\varepsilon$. Then, the following approximate 
anticommutation relations are true
\begin{equation}\label{Rob:lem1:1}
\|\{A_{1}',A_{\frac{n}{2}+1}'\}\ket{\psi'}\|\leq \sqrt{2\varepsilon_1}+\frac{1}{\xi_{n/2-1}}\left(\frac{4\sqrt{\epsilon_1}}{\alpha_{n/2-1}}+n\sqrt{2\varepsilon_2}\right) = \omega_{\mathrm{ev}}
\end{equation}
for even $n$, and 
\begin{eqnarray}
\label{Rob:lem1:1a}
\|\{A_1',A_{\frac{n+1}{2}}'+A_{\frac{n+3}{2}}'\}\ket{\psi'}\|&\leq& 2\sqrt{\varepsilon_1 n}\left(\frac{\sqrt{2}}{\zeta_{(n-1)/2}}+\sqrt{n-1}\right)
+\sqrt{\varepsilon_1}(1+\sqrt{2})
\nonumber\\
&&+\frac{3\sqrt{\varepsilon_1}}{\cos\case{\pi}{2n}\alpha_{(n-1)/2}\zeta_{(n-1)/2}}
\left(2+\frac{\gamma_{(n-1)/2}}{\alpha_1}\right) = \omega_{\mathrm{odd}}\nonumber\\
\end{eqnarray}
for odd $n$. For any fixed $n$ the right-hand sides of both inequalities vanish if 
$\varepsilon\to 0$ and for sufficiently large $n$ both functions scale quadratically with $n$.
\end{lem}
\begin{proof}The proof goes along the same lines as that of Lemma \ref{Exact:Anti}, however, at each step we need to take into account the error stemming from the fact that now the Bell inequality is not violated maximally. We prove the cases of 
even and odd $n$ separately.

\textit{Even $n$.} We first need to prove the approximate versions of the identities (\ref{Lemma1:5}) and (\ref{Lemma1:6}). By substituting $j=k-i$ in (\ref{eq:rob2})
we obtain
\begin{equation}\label{antirob1}
\|(\alpha_iC_{k-i}+\beta_i C_k+\gamma_iC_{k+1})\ket{\psi'}\|\leq \sqrt{\varepsilon_1}
\end{equation}
Then, by shifting $i\to n-i-1$ and setting $j=k+i+1$ in (\ref{eq:rob2}), 
we have
\begin{equation}\label{antirob2}
\|(\alpha_iC_{k+i+1}+\gamma_i C_k+\beta_iC_{k+1})\ket{\psi'}\|\leq \sqrt{\varepsilon_1}.
\end{equation}
Both inequalities imply 
\begin{equation}\label{antirob3}
\|(C_{k-i}+C_{k+i+1}-\xi_iC_{k,k+1})\ket{\psi'}\|\leq \frac{2\sqrt{\varepsilon_1}}{\alpha_i}
\end{equation}
for any $k =1,\ldots,n$ and $i =1,\ldots,n-2$. The case when $i = n-1$ or $i = n$ are trivial because they represent the definition of  $C_{k,k+1}$.
Then, by using Eq. (\ref{antirob1}) with $k=n$ and Eq. (\ref{antirob2}) with $k=0$, one can prove the following inequality
\begin{eqnarray}\label{antirob4}
\|(C_{i+1}-C_{n-i}-\xi_iC_{1,-n})\ket{\psi'}\|\leq \frac{2\sqrt{\varepsilon_1}}{\alpha_i}
\end{eqnarray}
with $i=1,\ldots,n-2$.
Now, one has 
\begin{eqnarray}
\|\{A_1',A_{\frac{n}{2}+1}'\}\ket{\psi'}\|&=&\|(A_1'A_{\frac{n}{2}+1}'+A_{\frac{n}{2}+1}'A_1')\ket{\psi'}\|\nonumber\\
&\leq &\|(A_1'B_{\frac{n}{2},\frac{n}{2}+1}'+A_{\frac{n}{2}+1}'B_{1,-n}')\ket{\psi'}\|+\sqrt{2\varepsilon_1},
\end{eqnarray}
which with the aid of Eq. (\ref{antirob3}) with $k=n/2$ and $i=n/2-1$
and Eq. (\ref{antirob4}) with $i=n/2-1$, can be further upper bounded as
\begin{eqnarray}\label{antiIneq0}
\fl\left\|\{A_1',A_{\frac{n}{2}+1}'\}\ket{\psi'}\right\|&\leq &\frac{1}{\xi_{n/2-1}}\left\|\left[A_1'(B_1'+B_n')+A_{\frac{n}{2}+1}'(B_{\frac{n}{2}}'-B_{\frac{n}{2}+1}')\right]\ket{\psi'}\right\|+\frac{1}{\xi_{n/2-1}}\frac{4\sqrt{\varepsilon_1}}{\alpha_{n/2-1}}\nonumber\\
&\leq&\frac{1}{\xi_{n/2-1}}\left[\left\|(A_1'B_1'-A_{\frac{n}{2}+1}'B_{\frac{n}{2}+1}')\ket{\psi'}\right\|+\left\|(A_1'B_n'+A_{\frac{n}{2}+1}'B_{\frac{n}{2}}')\ket{\psi'}\right\|\right]\nonumber\\
&&+\frac{1}{\xi_{n/2-1}}\frac{4\sqrt{\varepsilon_1}}{\alpha_{n/2-1}}.
\end{eqnarray}
To upper bound the above two terms, we will use approximate versions of 
Eqs. (\ref{Lemma1:9}) and (\ref{Lemma1:10}) First, it follows from the SOS decomposition that for any $j=1,\ldots,n$:
\begin{equation}
\sum_{i=1}^{j}\left\|(A_i'B_i'-A_{i+1}'B_{i+1}')\ket{\psi'}\right\|^2\leq n\varepsilon_2,
\end{equation}
which by virtue of the triangle inequality for the norm and concavity of the square root implies 
\begin{eqnarray}\label{antiIneq1}
\left\|(A_1'B_1'-A_{j+1}'B_{j+1}')\ket{\psi'}\right\|&=&\left\|\sum_{i=1}^{j}(A_i'B_i'-A_{i+1}'B_{i+1}')\ket{\psi'}\right\|\nonumber\\
&\leq &\sum_{i=1}^{j}\left\|(A_i'B_i'-A_{i+1}'B_{i+1}')\ket{\psi'}\right\|\nonumber\\
&\leq & \sqrt{j}\sqrt{\sum_{i=1}^{j}\left\|(A_i'B_i'-A_{i+1}'B_{i+1}')\ket{\psi'}\right\|^2}\nonumber\\
&\leq & \sqrt{jn\varepsilon_2}.
%
%2\sqrt{2}\sqrt{j}\sqrt{n\varepsilon\cos\case{\pi}{2n}}.
\end{eqnarray}
Analogously, the SOS decomposition (\ref{eq:SOSnBell2ndorder})
implies that  
\begin{equation}
\sum_{i=1}^{j}\left\|(A_i'B_{i-1}'-A_{i+1}'B_{i}')\ket{\psi'}\right\|^2\leq \sqrt{n\varepsilon_2},
\end{equation}
from which, by using similar arguments as above, one infers that 
\begin{eqnarray}\label{antiIneq2}
\left\|(A_1'B_n'+A_{j+1}'B_{j}')\ket{\psi'}\right\|=\sum_{i=1}^{j}\left\|(A_i'B_{i-1}'-A_{i+1}'B_{i}')\ket{\psi'}\right\|\leq\sqrt{jn\varepsilon_2}.
\end{eqnarray}
Substituting $j=n/2$ and applying both inequalities (\ref{antiIneq1}) and (\ref{antiIneq2}) to (\ref{antiIneq0})
one finally obtains (\ref{Rob:lem1:1}).

\textbf{Odd number of measurements.} We first prove the following inequality
\begin{equation}\label{antiIneq3}
\|(C_{k-i}+C_{k+i}-\zeta_iC_k)\ket{\psi'}\|\leq \left(2+\frac{\gamma_i}{\alpha_1}\right)\frac{\sqrt{\varepsilon_1}}{\alpha_i}
\end{equation}
for any $i =1,\ldots,n-2$. Then, from inequalities (\ref{antiIneq1}) and (\ref{antiIneq2}) with $j=(n-1)/2$, 
and inequality (\ref{antiIneq3}) for $i=(n-1)/2$ and $k=n$, one obtains
\begin{eqnarray}
\left\|\left[A_1'(B_1'+B_n')-\zeta_{\frac{n-1}{2}}A_{\frac{n+1}{2}}'B_n'\right]\ket{\psi'}\right\|&\leq& \sqrt{2n(n-1)\varepsilon_2}+\varepsilon',
%
%\frac{\sqrt{\varepsilon_1}}{\alpha_{(n-1)/2}}\left(2+\frac{\gamma_{(n-1)/2}}{\alpha_1}\right).\nonumber\\
\end{eqnarray}
where we denoted 
\begin{equation}
\varepsilon'=\frac{\sqrt{\varepsilon_1}}{\alpha_{(n-1)/2}}\left(2+\frac{\gamma_{(n-1)/2}}{\alpha_1}\right).
\end{equation}
Analogously, from inequalities (\ref{antiIneq1}) and (\ref{antiIneq2}) with $j=(n+1)/2$ 
and inequality (\ref{antiIneq3}) for $i=(n-1)/2$ and $k=n+1$, one obtains
\begin{eqnarray}\label{Salavrakos}
\left\|\left[A_1'(B_1'+B_n')+\zeta_{\frac{n-1}{2}}A_{\frac{n+3}{2}}'B_n'\right]\ket{\psi'}\right\|&\leq& \sqrt{2n(n-1)\varepsilon_2}+\varepsilon'.
%
%\frac{\sqrt{\varepsilon_1}}{\alpha_{(n-1)/2}}\left(2+\frac{\gamma_{(n-1)/2}}{\alpha_1}\right).\nonumber\\
\end{eqnarray}
We can then upper bound 
\begin{eqnarray}\label{antiIneq4}
\fl\left\|\{A_1',A_{\frac{n+1}{2}}'+A_{\frac{n+3}{2}}'\}\ket{\psi'}\right\|&\leq& \frac{1}{2\cos(\case{\pi}{2n})}
\left\|[A_1'(B_{\frac{n-1}{2}}'+2B_{\frac{n+1}{2}}'+B_{\frac{n+3}{2}}')\right.\nonumber\\
&&\left.\hspace{2cm}+(A_{\frac{n+1}{2}}'+A_{\frac{n+3}{2}}')(B_1'-B_n')]\ket{\psi'}\right\|+\sqrt{\varepsilon_1}(1+\sqrt{2})\nonumber\\
&&\hspace{-3cm}\leq \frac{1}{2\cos(\case{\pi}{2n})}\left\|\left[A_1'\left(B_{\frac{n-1}{2}}'+B_{\frac{n+3}{2}}'+2\frac{B_1'+B_n'}{\zeta_{(n-1)/2}}\right)+(A_{\frac{n+1}{2}}'+A_{\frac{n+3}{2}}')(B_1'-B_n')\right]\ket{\psi'}\right\|\nonumber\\
&&\hspace{-2.5cm}+\sqrt{\varepsilon_1}(1+\sqrt{2})+\frac{\varepsilon'}{2\cos(\pi/2n)\zeta_{(n-1)/2}}\nonumber\\
&&\hspace{-3cm}\leq\frac{1}{2\cos(\case{\pi}{2n})}\left\|\left(A_1'B_{\frac{n-1}{2}}'+A_1'B_{\frac{n+3}{2}}'+A_{\frac{n+1}{2}}'B_1'-A_{\frac{n+3}{2}}'B_n'\right)\ket{\psi'}\right\|\nonumber\\
&&\hspace{-2.5cm}+\sqrt{\varepsilon_1}(1+\sqrt{2})+\frac{3\varepsilon'}{2\cos(\pi/2n)\zeta_{(n-1)/2}}
+2\sqrt{\varepsilon_1 n(n-1)}.
\end{eqnarray}
In the first inequality we used  (\ref{eq:rob1}) twice in parallel (to exchange $A_{n+2}'$ and $A_{n+3}'$ with corresponding $B'$s) and once more separately (to exchange $A_1'$ with $B_{1,-n}'$). To get the second inequality we used (\ref{antiIneq3}) and for the final inequality we used twice (\ref{Salavrakos}).
Inequality (\ref{antiIneq3}) for $k=1$ and $i=(n-1)/2$ gives 
\begin{eqnarray}
\left\|[A_{\frac{n+1}{2}}'-A_{\frac{n+3}{2}}'-\zeta_{\frac{n-1}{2}}A_1']\ket{\psi'}\right\| 
\leq\varepsilon',\\ \left\|[B_{\frac{n+1}{2}}'-B_\frac{n+3}{2}'-\zeta_{\frac{n-1}{2}}B_1']\ket{\psi'}\right\|
\leq\varepsilon' 
%
%\left(2+\frac{\gamma_{(n-1)/2}}{\alpha_1}\right)\frac{\sqrt{\varepsilon_1}}{\alpha_{(n-1)/2}}
\end{eqnarray}
%
%\begin{equation}
%\|[B_{\frac{n+1}{2}}-B_\frac{n+3}{2}-\zeta_{\frac{n-1}{2}}B_1]\ket{\psi'}\|
%\leq \varepsilon' 
%
%\left(2+\frac{\gamma_{(n-1)/2}}{\alpha_1}\right)\frac{\sqrt{\varepsilon_1}}{\alpha_{(n-1)/2}},
%\end{equation}
%
with $C=A,B$, while for $k=n$ and $i=(n-1)/2$
\begin{equation}
\left\|[B_{\frac{n+1}{2}}'-B_\frac{n-1}{2}'-\zeta_{\frac{n-1}{2}}B_n']\ket{\psi'}\right\|
\leq\varepsilon', 
%
%\left(2+\frac{\gamma_{(n-1)/2}}{\alpha_1}\right)\frac{\sqrt{\varepsilon_1}}{\alpha_{(n-1)/2}},
\end{equation}
These three inequalities when applied to (\ref{antiIneq4}) give
\begin{eqnarray}
\fl\left\|\{A_1',A_{\frac{n+1}{2}}'+A_{\frac{n+3}{2}}'\}\ket{\psi'}\right\|&\leq& \frac{1}{2\cos(\case{\pi}{2n})\zeta_{\frac{n-1}{2}}}
\left\|\left(A_{\frac{n+1}{2}}'B_{\frac{n-1}{2}}'-A_{\frac{n+3}{2}}'B_{\frac{n+1}{2}}'
\right.\right.\nonumber\\
&&\hspace{3cm}\left.\left.+A_{\frac{n+1}{2}}'B_{\frac{n+1}{2}}'-A_{\frac{n+3}{2}}'B_{\frac{n+3}{2}}'\right)\ket{\psi'}\right\|\nonumber\\
&&+\sqrt{\varepsilon_1}(1+\sqrt{2})+\frac{3\varepsilon'}{\cos\case{\pi}{2n}\zeta_{(n-1)/2}}+2\sqrt{\varepsilon_1n(n-1)}.
\end{eqnarray}
To upper bound the norm appearing on the right-hand side and complete the proof we 
use inequalities (\ref{eq:rob3}) and (\ref{eq:rob4}) with $i=(n+1)/2$ which leads us to 
\begin{eqnarray}
\left\|\{A_1',A_{\frac{n+1}{2}}'+A_{\frac{n+3}{2}}'\}\ket{\psi'}\right\|&\leq& 2\sqrt{\varepsilon_1 n}\left(\frac{\sqrt{2}}{\zeta_{(n-1)/2}}+\sqrt{n-1}\right)
+\sqrt{\varepsilon_1}(1+\sqrt{2})
\nonumber\\
&&+\frac{3\sqrt{\varepsilon_1}}{\cos\case{\pi}{2n}\alpha_{(n-1)/2}\zeta_{(n-1)/2}}
\left(2+\frac{\gamma_{(n-1)/2}}{\alpha_1}\right).
\end{eqnarray}
To complete the proof let us notice that both $\omega_{ev}$ and $\omega_{odd}$, defined in Eqs. (\ref{Rob:lem1:1}) and (\ref{Rob:lem1:1a}) respectively, vanish when $\varepsilon \to 0$. Furthermore, the term dominating the scaling of $\omega_{ev}$ with $n$ for large $n$ is $4\varepsilon_1/(\xi_{n/2-1}\alpha_{n/2-1}) = 2\sqrt{\varepsilon}/(\sin^2(\pi/2n))$. It follows that for sufficiently large $n$ the function $1/\sin^2(\pi/2n)$ behaves like $(4/\pi^2)n^2+1/3+O(1/n^2)$ and therefore we can conclude that $\omega_{ev}$ scales quadratically with $n$ when $n$ is large enough, and for small $\varepsilon$ it behaves as $\sqrt{\varepsilon}$. After analogous analysis one finds that $\omega_{odd}$ exhibits the same behaviour for small $\varepsilon$ and sufficiently large $n$.
\end{proof}
\begin{lem}\label{lastlemma}
Let $\ket{\psi'}$ and $A_i',B_i'$ be a state and measurements violating the chained Bell inequalities by 
$B_n^{\max}-\varepsilon$. Then, for an even number of measurements:
\begin{eqnarray}
&&\left\|\left(A_i'-s_iA_{\frac{n}{2}+1}'-c_iA_1'\right)\ket{\psi'}\right\|\leq g_{\mathrm{ev}}(\varepsilon,n),\nonumber\\
&&\left\|\left(B_i'-s_i'B_{\frac{n}{2},\frac{n}{2}+1}' -c_i'B_{1,-n}'\right)\ket{\psi'}\right\|\leq h_{\mathrm{ev}}(\varepsilon,n),
\end{eqnarray} 
while for an odd number of measurements:
\begin{eqnarray}
&&\left\|\left(A_i'-s_iA_{\frac{n+1}{2},\frac{n+3}{2}}' -c_iA_1'\right)\ket{\psi'}\right\|\leq g_{\mathrm{odd}}(\varepsilon,n),\nonumber\\
&&\left\|\left(B_i'-s_i'B_{\frac{n+1}{2}}' -c_i'B_{1,-n}'\right)\ket{\psi'}\right\|\leq h_{\mathrm{odd}}(\varepsilon,n).
\end{eqnarray}
The functions $g_{\mathrm{ev}}$, $h_{\mathrm{ev}}$, 
$g_{\mathrm{odd}}$ and $h_{\mathrm{odd}}$ vanish for $\epsilon\to 0$ and scale linearly with $n$.
\end{lem}
\begin{proof}
We will follow the proof of Lemma \ref{Structure}. We can write
\begin{eqnarray}\label{Skrzpczk1}
& \left\|\left(A_i'-s_iA_{\frac{n}{2}+1}'
-c_iA_1'\right)\ket{\psi'}\right\|\
  \nonumber\\ & = \frac{1}{2}\left\|\left(A_i' - A_{2-i}' - \zeta_{\frac{n}{2}+1-i}A_{\frac{n}{2}+1}' + A_i' + A_{2-i}' - \zeta_{i-1}A_{1}'\right)\ket{\psi'}\right\| \nonumber\\
  & \leq \frac{1}{2}\left\|\left(A_i' - A_{2-i}' - \zeta_{\frac{n}{2}+1-i}A_{\frac{n}{2}+1}'\right) \ket{\psi'}\right\|\ + \frac{1}{2}\left\|\left(A_i' + A_{2-i}' - \zeta_{i-1}A_{1}'\right)\ket{\psi'}\right\| \nonumber\\
  & \leq \left(1 + \frac{\gamma_{|\frac{n}{2}+1-i|}}{2\alpha_{1}}\right)\frac{\sqrt{\varepsilon_1}}{\alpha_{|\frac{n}{2}+1-i|}} + \left(1 + \frac{\gamma_{i-1}}{2\alpha_{1}}\right)\frac{\sqrt{\varepsilon_1}}{\alpha_{i-1}} = g_{\mathrm{ev}}
\end{eqnarray}
The equality is just rewritten pair of Eqs. (\ref{Lemma2:1}) and (\ref{Lemma2:2}), the first inequality is the triangle inequality followed by the bounds from Eq. (\ref{antiIneq3}). Absolute value appearing in $\gamma_{|\frac{n}{2}+1-i|}$ and $\alpha_{|\frac{n}{2}+1-i|}$ is justified in the discussion after Eq. (\ref{Lemma2:1}). Note that this bound cannot be applied to the cases when $i = 1,n/2 +1, n$ because for these cases coefficients $\alpha_i$ and $\gamma_i$ are not defined. The cases $i = 1,n/2 +1$ are trivial statements and $g_{\mathrm{ev}} = 0$, while for the case $i = n$ the norm $\left\|\left(A_i' + A_{2-i}' - \zeta_{i-1}A_{1}'\right)\ket{\psi'}\right\| \leq \sqrt{\varepsilon_1/\alpha_1}$ is obtained by fixing $j = n$ and $i=1$ in (\ref{eq:rob2}), so $g_{\mathrm{ev}} = (1 + \gamma_{|\frac{n}{2}+1-i|}/2\alpha_{1})(\sqrt{\varepsilon_1}/\alpha_{|\frac{n}{2}+1-i|}) + \sqrt{\varepsilon_1/\alpha_1}/2$. Similarly it can be shown that:
\begin{eqnarray}\label{Skrzpczk2}
& \left\|\left(B_i'-s_i'B_{\frac{n}{2},\frac{n}{2}+1}'
-c_i'B_{1,-n}'\right)\ket{\psi'}\right\|\
  \nonumber\\ & = \frac{1}{2}\left\|\left(B_i' - B_{1-i}' - \xi_{\frac{n}{2}-i}B_{\frac{n}{2},\frac{n}{2}+1}' + B_i' + B_{1-i}' - \xi_{i-1}B_{1,-n}'\right)\ket{\psi'}\right\| \nonumber\\
  & \leq \frac{1}{2}\left\|\left(B_i' - B_{1-i}' - \xi_{\frac{n}{2}-i}B_{\frac{n}{2},\frac{n}{2}+1}'\right)\ket{\psi'}\right\|\ + \frac{1}{2}\left\|\left(B_i' + B_{1-i}' - \xi_{i-1}B_{1,-n}'\right)\ket{\psi'}\right\| \nonumber\\
  &\leq \sqrt{\varepsilon_1}\left( \frac{1}{\alpha_{i-1}} + \frac{1}{\tilde{\alpha}_{\frac{n}{2}-i}}\right) = h_{\mathrm{ev}},
\end{eqnarray}
where in the last inequality we used already established bounds given in Eqs. (\ref{antirob3}) and (\ref{antirob4}) and we introduced notation $\tilde{\alpha}_{n/2-i}$ which is equal to $\alpha_{n/2-i}$ when $n/2 > i$, and to $\alpha_{i - 1 -n/2}$ otherwise (for the clarification see the text following Eq. (\ref{Lemma2:3})). Similarly to the previous case the bound is properly defined unless $i \in \{1,n,n/2, n/2+1 \}$. For the cases $i = 1,n$ the norm $\Vert(B_i' + B_{1-i}' - \xi_{i-1}B_{1,-n}')\ket{\psi'}\Vert $ is trivial, thus equal to $0$, so we have $h_{\mathrm{ev}} = \sqrt{\varepsilon_1}/\tilde{\alpha}_{n/2-i}$. Similarly when $i = n/2, n/2+1$, the norm $\Vert(B_i' - B_{1-i}' - \xi_{\frac{n}{2}-i}B_{\frac{n}{2},\frac{n}{2}+1}')\ket{\psi'}\Vert $ is equal to $0$, causing $h_{\mathrm{ev}}$ to be equal to $\sqrt{\varepsilon_1}/\alpha_{i-1}$.
By repeating analogue procedure it is easy to obtain bounds for the case when the number of inputs is odd:
\begin{eqnarray}
& g_{\mathrm{odd}} = \sqrt{\varepsilon_1}\left(\frac{1}{\tilde{\alpha}_{\frac{n+1}{2}-i}}+\left(1 + \frac{\gamma_{i-1}}{2\alpha_1}\right)\frac{1}{\alpha_{i-1}}  \right), \\
& h_{\mathrm{odd}} = \sqrt{\varepsilon_1}\left(\frac{1}{\alpha_{i-1}}+\left(1 + \frac{\gamma_{|\frac{n+1}{2}-i|}}{2\alpha_1}\right)\frac{1}{\alpha_{|\frac{n+1}{2}-i|}}  \right).
\end{eqnarray}
Similarly to the case when the number of inputs is even for $i = 1,n$ the expression for $g_{\mathrm{odd}}$ is estimated to be $\sqrt{\varepsilon_1}/ \tilde{\alpha}_{\frac{n+1}{2}-i}$ and for $i = (n+1)/2,(n+3)/2$ it reduces to $\left[\sqrt{\varepsilon_1}/ \alpha_{i-1}\right]\left(1 + \gamma_{i-1}/(2\alpha_1)\right)$. Also, for $i = (n+1)/2$ we have $h_{\mathrm{odd}} = \sqrt{\varepsilon_1}/\alpha_{i-1}$,and for $i = 1,n$ we estimate  $h_{\mathrm{odd}}=[\sqrt{\varepsilon_1}/ \alpha_{|\frac{n+1}{2}-i|}](1 + \gamma_{|\frac{n+1}{2}-i|}/(2\alpha_1))$.\\
In the worst case functions $g_{\mathrm{ev}}$,$h_{\mathrm{ev}}$,$g_{\mathrm{odd}}$ and $h_{\mathrm{odd}}$ behave as $\sin^{-1}(\pi/n)$ when $n$ is sufficiently large. Linear scaling with respect to $n$ of the aforementioned functions when $n$ is sufficiently large can be confirmed by considering the behaviour of function $\sin^{-1}(\pi/n)$ when $n$ is large enough.
\end{proof}

\begin{lem}\label{junklemma}
Let $\ket{\varphi}$ be the state of the additional degrees of freedom from Theorem \ref{thm1} and $\ket{\varphi'}$ state defined in Eq. (\ref{eq:phi'}). Then,
\begin{equation}\label{cvrkut}
\| \ket{\varphi} - \ket{\varphi'}\| \leq \left(\frac{1}{2} + \sqrt{2} \right)\sqrt{\varepsilon_1} + \frac{\omega'}{4},
\end{equation}
where $\omega'\equiv\omega_{\mathrm{ev}}$ for even $n$ and $\omega'\equiv\omega_{\mathrm{odd}}$ for odd $n$.
\end{lem}
\begin{proof}
Let us notice that 
$\|\ket{\varphi}-\ket{\varphi'}\|=\|\ket{\varphi'}\|-1$ and then by using the explicit form of $\ket{\varphi'}$ and the inequalities (\ref{nierownosc1}) and (\ref{nierownosc2}), we can write
\begin{eqnarray}
\label{cvrkut0}
\| \ket{\varphi'} \| & \leq & \frac{1}{2\sqrt{2}} \left( \| (\mathbbm{1}+Z_A')(\mathbbm{1}+Z_B')\ket{\psi'} \| + 2\sqrt{\varepsilon_1} \right) \nonumber\\
& \leq &\frac{1}{2\sqrt{2}} \left[ \| (\mathbbm{1}+Z_A')^2\ket{\psi'} \| + 4\sqrt{\varepsilon_1} \right] \nonumber\\ 
& = &\frac{1}{\sqrt{2}}\| (\mathbbm{1}+Z_A')\ket{\psi'} \| + \sqrt{2\varepsilon_1}
\end{eqnarray}
Now we want to estimate $|\ket{\psi'}Z_A'\ket{\psi'}|$. For this we will follow similar estimation presented in \cite{MKYS}. Note that due to unitarity of $Z_A'$, and Eqs. (\ref{nierownosc1}) and (\ref{nierownosc3})  we can write $\|\ (Z_A'X_B' +X_A'Z_A')\ket{\psi'}\|\ =  \|\ (Z_A'X_B'-Z_A'X_A'+Z_A'X_A' +X_A'Z_A')\ket{\psi'}\|\ \leq \sqrt{\varepsilon_1} + \omega'$. The norm will not change if we multiply the expression in brackets by some unitary operator. This means that $|\bra{\psi'}Z_A'\ket{\psi'} + \bra{\psi'}X_B'X_A'Z_A'\ket{\psi'}| \leq \sqrt{\varepsilon_1} + \omega'$. We can put the same bound for the complex conjugated expression
\begin{equation}\label{cvrkut1}
 |\bra{\psi'}Z_A'\ket{\psi'} + \bra{\psi'}X_B'Z_A'X_A'\ket{\psi'}| \leq \sqrt{\varepsilon_1} + \omega'.
\end{equation} 
 On the other hand, using unitarity of $\bra{\psi'}Z_A'$ and result (\ref{nierownosc1}) we can write
 \begin{equation}\label{cvrkut2}
 |\bra{\psi'}Z_A'\ket{\psi'} - \bra{\psi'}X_B'Z_A'X_A'\ket{\psi'}| \leq \sqrt{\varepsilon_1}.
\end{equation} 
Finally if we sum Eqs. (\ref{cvrkut1}) and (\ref{cvrkut2}) we get 
\begin{equation}
\label{cvrkut3}
|\bra{\psi'}Z_A'\ket{\psi'}| \leq \sqrt{\varepsilon_1} + \omega'/2
\end{equation}
If we plug this result in (\ref{cvrkut0}) we will get
\begin{eqnarray}
\|\ \ket{\varphi'} \|\ & \leq \sqrt{\bra{\psi'}(\mathbbm{1}+Z_A')\ket{\psi'}} + \sqrt{2\varepsilon_1} \nonumber\\ & \leq \sqrt{1 + \sqrt{\varepsilon_1} + \omega'/2} + \sqrt{2\varepsilon_1} \nonumber\\
& \leq 1 + (\frac{1}{2}+\sqrt{2})\sqrt{\varepsilon_1} + \frac{\omega'}{4}
\end{eqnarray}
This estimation concludes the proof, since it is easy to check that the Eq. (\ref{cvrkut}) is satisfied.
\end{proof}
\end{document}